\documentclass[10pt,journal,cspaper,compsoc]{IEEEtran}
\usepackage{amsmath}
\usepackage{latexsym}
\usepackage{comment}
\usepackage{enumerate}
\usepackage{amssymb}
\usepackage{ifpdf}
    \ifpdf
      \usepackage[pdftex]{graphicx}
      \usepackage{epstopdf}
    \else
      \usepackage{graphicx}
    \fi
\usepackage{cite}
\newtheorem{theorem}{Theorem}
\newtheorem{lemma}{Lemma}
\newtheorem{proposition}{Proposition}

\newtheorem{definition}{Definition}

\title{Spread of Influence and Content in Mobile Opportunistic Networks}
\begin{document}

\author{Srinivasan~Venkatramanan,~\IEEEmembership{Student Member,~IEEE,}
        Anurag~Kumar,~\IEEEmembership{Fellow,~IEEE,}\protect\\
        Department of Electrical Communication Engineering, Indian Institute of Science, \protect\\ Bangalore - 560012, India.\protect\\
        E-mail: vsrini,anurag@ece.iisc.ernet.in
\IEEEcompsocitemizethanks{\IEEEcompsocthanksitem Some of the content of this paper was presented in the IEEE Infocom 2012 mini-conference \cite{srini-kumar12coevol-mp2p}. This work was supported by
IFCPAR (Indo-French Centre for the Promotion of Advanced Research) under Project 4000-ITA, and by the Department of Science and Technology (DST), Govt. of India, through a J.C. Bose Fellowship. 
}
\thanks{}}

% \author{
% \IEEEauthorblockN{Srinivasan Venkatramanan and Anurag Kumar}
% \IEEEauthorblockA{Department of Electrical Communication Engineering\\
% Indian Institute of Science\\
% Bangalore 560012, India\\
% Email: vsrini,anurag@ece.iisc.ernet.in}
% }

\IEEEcompsoctitleabstractindextext{
\begin{abstract}
We consider a setting in which a single item of content (such as a song or a video clip) is disseminated in a population of mobile nodes by opportunistic copying when pairs of nodes come in radio contact. We propose and study models that capture the joint evolution of the population of nodes interested in the content (referred to as \emph{destinations}), and the population of nodes that possess the content. The evolution of interest in the content is captured using an influence spread model and the content spread occurs via epidemic copying. Nodes not yet interested in the content are called relays; the influence spread process converts relays into destinations. We consider the decentralized setting, where interest in the content and the spread of the content evolve by pairwise interactions between the mobiles. We derive fluid limits for the joint evolution models and obtain optimal policies for copying to relay nodes in order to deliver content to a desired fraction of destinations. We prove that a time-threshold policy is optimal while copying to relays. We then provide insights into the effects of various system parameters on the co-evolution model through simulations.
\end{abstract}
}
\maketitle
\section{Introduction}
Due to the ubiquity of cellular networks, there has been a proliferation of hand-held mobile devices. The idea of mobile opportunistic networking is to exploit the mobility of users carrying such devices to transfer messages directly to each other during chance meetings. This is enabled by low-power radio interfaces on these devices (such as Bluetooth), and provides the opportunity for creating a multi-hopping communication network completely bypassing the cellular infrastructure. Since such a scheme cannot meet delay guarantees, applications that utilize mobile opportunistic networking must be \emph{delay tolerant}. On the other hand, such a peer-to-peer (P2P) content delivery is scalable \cite{repantis-kalogeraki04data-dissem-mp2p}, as the rate of service scales in proportion to the number of peers in the system with little additional cost to the system planner. Such a combination of an application and opportunistic transport is also called \emph{delay tolerant networking (DTN).} 

In this paper, we consider such a networking paradigm, and study the problem of dissemination of an item of content among the population of mobile nodes. The content could be a video (news footage, sports highlights, movie teaser, etc.) or an audio file (a recent song, a popular ringtone, etc.) The individuals carrying the mobile devices also interact socially (thus forming a \emph{social network}) and can \emph{influence} each other to become interested in the content. In such a situation, the system objective could be to facilitate the spread of the content to as many interested nodes as possible. In doing this, nodes that are not (yet) interested in the content could be used to relay the content, thus making the content more available. The overall objective would then be to spread the content among those who are interested, while limiting the number of copies to those nodes who are not yet interested in the content.

\noindent\emph{Literature Survey:} In this paper, we refer to nodes who are interested in the content as \emph{destinations} and those who are not yet interested in the content as \emph{relays}.  In the delay tolerant networking context, there is prior literature (see \cite{singh-etal11dtn-multi-destination} and \cite{khouzani-etal10patch-dissemination} and the references therein) on the optimal opportunistic copying of content in order to optimize delivery delay and/or wasteful copying to relays. 

An interesting aspect, that is largely unexplored in prior literature, is the evolution of interest in the content. For recently ``released'' content (such as a new video clip), the fraction of nodes interested in it (destinations), will not be fixed, but will evolve based on the influence of nodes that are already destinations. Such influence could be mediated via a centralized server, which uses a low bit-rate channel to broadcast the current popularity of the content, or by interactions between the mobile nodes themselves. Thus, it may be necessary to deliver the content to destinations while keeping track of the demand evolution. 

A recent work that explores this aspect \cite{shakkottai-etal10demand-aware-content-spread}, uses epidemic models to characterize demand evolution and aims to obtain a hybrid of P2P and client-server architecture to efficiently meet the demand. While \cite{shakkottai-etal10demand-aware-content-spread} begins by assuming the Bass model for interest evolution, in our work, we analyze the interest evolution in its own right, adopting the Linear Threshold (LT) model \cite{kempe-etal03max-spread-infl} from the viral marketing literature and deriving its fluid limit. Also, while in \cite{shakkottai-etal10demand-aware-content-spread} the peer-to-peer file dissemination occurs only among the nodes interested in the content, we consider the notion of relays (as yet uninterested in the content) aiding the spread, thus leading to a more general dynamics of the co-evolution model.

\noindent \emph{Contributions:} In this paper we consider a population of $N$ mobile nodes, and model their pair-wise meetings by independent Poisson point processes. The item of content is provided to a subset of the initially interested set of nodes. The LT model is used to model the spread of interest between nodes. A controlled epidemic spread model is used for the opportunistic copying of the content between nodes. In this framework, we make the following contributions:

\begin{enumerate}
 \item We consider the \emph{homogeneous influence linear threshold model (HILT)} model which is a special case of the LT model introduced in Kempe et al.\ \cite{kempe-etal03max-spread-infl}. The population processes in the HILT model are modeled as a Markov chain, and a fluid limit is developed under a certain scaling of the spread dynamics as $N \to \infty$. The fluid limit nicely captures the influence threshold distribution via its hazard function. The well known SIR epidemic model\cite{kermack-mckendrick27SIR-epidemics} emerges as a special case when the threshold distribution is exponential. It is this SIR model that is then used in the remainder of the paper for modeling the spread of interest in the content. 

 \item For the case in which influence is spread by interaction between the devices, and a controlled epidemic spread of the content, we obtain the SIR-SI model, for which we obtain the fluid limit for a fixed probability, $\sigma$, of copying to a relay node (i.e., an uninterested node). 

 \item When the copying probability, $\sigma$, can vary with time, we obtain a controlled o.d.e., for which we obtain the optimal control by direct arguments using certain monotonicity properties. This results in a time-threshold structure of the optimal control. We provide an extensive numerical study of this model, thus providing additional interesting insights. 

\end{enumerate}

\noindent\emph{Outline:} In Section~\ref{sec:HILT-model}, we study the HILT model for evolution of interest, derive its fluid limit for general influence threshold distributions and discuss the effect of threshold distribution. In Section~\ref{sec:sir-si} we study the decentralized model (SIR-SI model) for co-evolution of content popularity and availability and establish the optimality of a time-threshold policy, for copying to relays. Finally, in Section~\ref{sec:numerical}, we first solve certain optimization problems concerning the evolution of interest. We then numerically compute optimal policies for the decentralized model, and study the effect of system parameters on the time threshold. 

\section{Interest Evolution: The HILT Model}
\label{sec:HILT-model}
In this section we introduce the homogeneous influence linear threshold (HILT) model used to model the evolution of interest in the content. In the original Linear Threshold (LT) model \cite{kempe-etal03max-spread-infl}, the individuals are modeled as nodes of a weighted directed graph $\mathcal{G}=(\mathcal{N},\mathcal{E})$, where $\mathcal{E} \subseteq \mathcal{N} \times \mathcal{N}$. With each ${i,j} \in \mathcal{E}$, there is associated a weight $w_{i,j}$ which gives a measure of \emph{influence} of node $i$ on node $j$, normalized such that the total weight into any node is at most 1, i.e., $\sum_{i} w_{i,j} \leq 1$. The Homogeneous Influence Linear Threshold (HILT model) \cite{srini-kumar11LT-model-ncc} is a special case of the LT model where the network graph is complete and all influence weights are equal. Hence, we have a mesh network on the population $\mathcal{N}$ containing $N=|\mathcal{N}|$ nodes with each edge carrying the same influence weight $\gamma_N=\frac{\Gamma}{N-1}$ and $\Gamma \leq 1$ (see Figure~\ref{fig:hilt}). 

\begin{figure}
\centering
\includegraphics[scale=0.3]{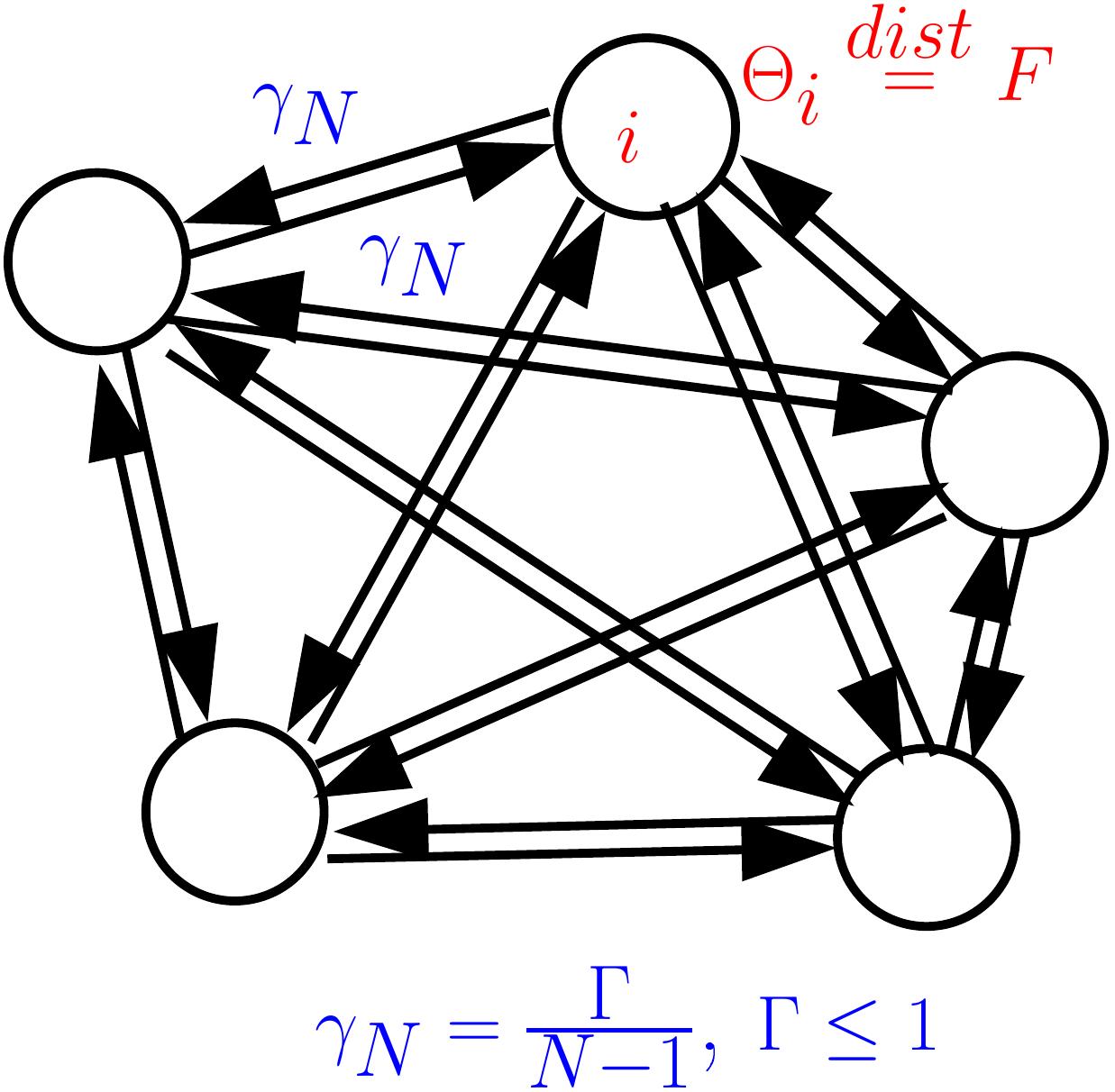}
\vspace{-0.3cm}
\caption{The graph for the HILT influence model}
\label{fig:hilt}
\vspace{-0.3cm}
\end{figure}
In the context of our content spread problem, at any step of influence spread, the population is partitioned into \emph{destination} nodes (i.e., the nodes that \emph{want} the item of content), and \emph{relays} (i.e., the nodes that, as yet, \emph{do not want} the content).  Given an initial set of destinations, $\mathcal A(0)$, we denote by $\mathcal{A}(k), k \in \{0, 1, 2, \cdots \}$ the process of the set of destinations. Each node $j \in \mathcal{N}$ independently chooses a random threshold $\Theta_{j} \geq 0$ from a given distribution $F(\cdot)$ \emph{at the beginning} and stays at this value thereafter \cite{kempe-etal03max-spread-infl}.  In the HILT model, the net influence of a set of destinations on any relay is $\gamma_N$ times the size of the destination set. 
Given the initial set of destinations and the thresholds sampled by all the relays, a relay gets converted to a destination when the total influence on it exceeds its threshold. In the distributed setting of a mobile opportunistic network, the influence can be exerted in one of two ways: 
\begin{itemize}
 \item Centralized spread: The system broadcasts the number of destinations over a low bit rate channel, thereby causing influence to be exerted on all the relays.
\item Distributed spread: Influence is exerted during the random meetings between destinations and relays.
\end{itemize}

In this section we develop a fluid limit for the HILT influence spread process, which will then motivate the distributed influence spread model that we will use in the remainder of the paper for our study of joint spread of influence and content. In the process of obtaining this model, we also provide the general fluid limit for the HILT process, and insights into the effect of the threshold distribution. 

We assume the \emph{progressive case}, i.e., conversion into destinations is an irreversible process, i.e., a relay $j \notin \mathcal{A}(k-1)$ gets influenced in step $k$ if, 
\begin{eqnarray}
\label{eqn:HILT-activation}
 \gamma_N | \mathcal{A}(k-1) |  \geq \Theta_{j} 
\end{eqnarray}
thereby converting to a destination, and then stays in that state from then on. The influence on a relay node can be viewed as building up cumulatively over time. The initial set $\mathcal{A}(0)$ is viewed as being \emph{infectious} and results in some of the relay nodes ``tipping'' over their thresholds and getting converted to destinations in the first period. A node $j$ that remains a relay  has the level of cumulative influence on it raised to $\gamma_N |\mathcal{A}(0)| (< \Theta_j)$. The nodes in $\mathcal{A}(0)$ are now viewed as being non-infectious, and the newly infected nodes are denoted by $\mathcal{D}(1)$, with $\mathcal{A}(1) = \mathcal{A}(0) + \mathcal{D}(1)$; see Figure~\ref{fig:spread_influence}. Thus, at the end of each period the population will contain three types of nodes: the set of destinations $\mathcal{A}(k)$, the set of newly infected destinations $\mathcal{D}(k)\subseteq \mathcal{A}(k)$ and the set of relays, $\mathcal{S}(k)$. Evidently, the activation process stops at a random step $T$ when there are no more infectious destinations, i.e., $\mathcal{D}(T) = \emptyset$, at this step the \emph{terminal set} $\mathcal{A}(T)$ is reached. 
\begin{figure} 
\centering
\includegraphics[scale=0.45]{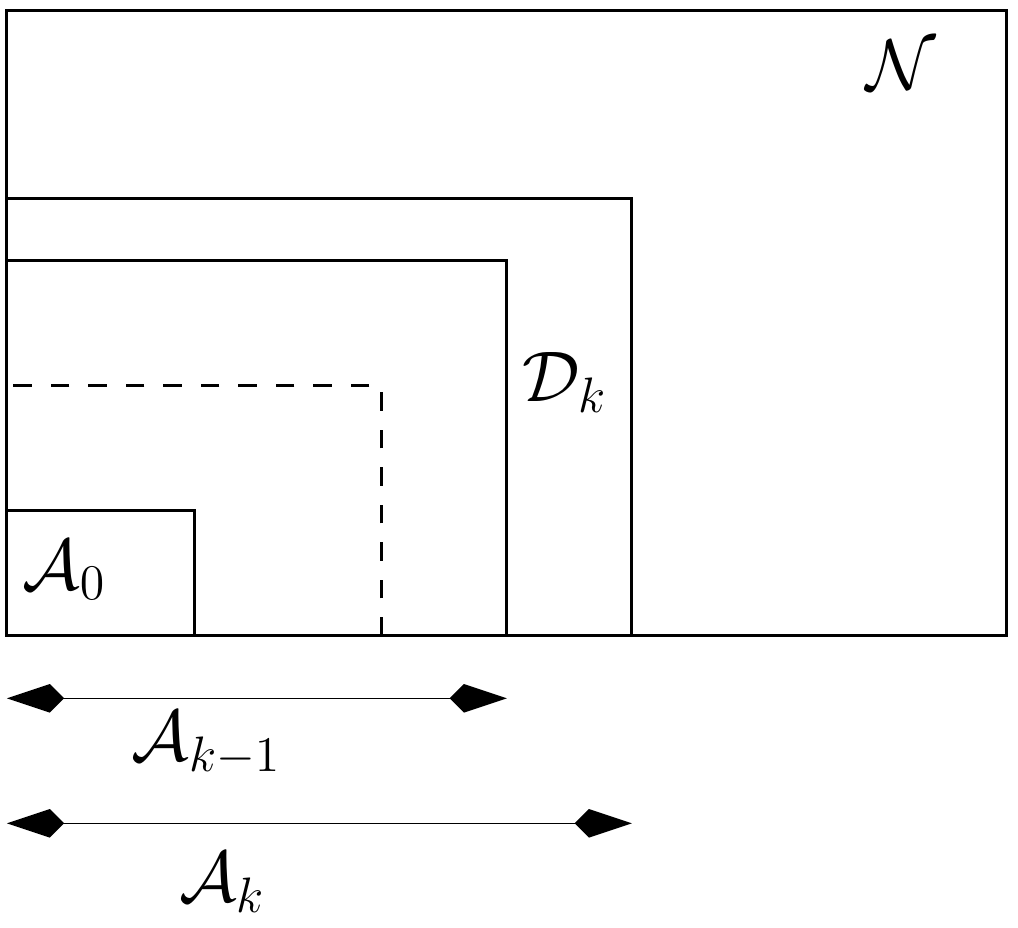}
\vspace{-0.3cm}
\caption{Spread of Influence in the HILT Model}
\label{fig:spread_influence}
\vspace{-0.5cm}
\end{figure}

\noindent\emph{Need for fluid limit:} From \cite{kempe-etal03max-spread-infl} it is known that the problem of identifying the most influential nodes in a social network under the Linear Threshold influence model is NP-hard. This is due to the difficulty in computing the expected influence (size of the terminal set) for a given seed set, which involves listing down all the path probabilities in the graph. In \cite{srini-kumar11LT-model-ncc}, we have provided an analytical expression to characterize the expected influence of a set. Using this for the HILT-model we can provide, a closed form expression for the expected size of the \emph{terminal destination set} $\mathcal{A}(T)$. But in order to exploit the interest evolution better, we would be interested in knowing the expected trajectory of the size of the destination set $\mathcal{A}(k)$, and we use fluid limits to characterize this expected evolution of the process. 

\subsection{O.D.E.\ Model for Interest Evolution}
\label{sec:HILT-fluid}
In this section we will use the convergence of a scaled discrete time Markov chain (DTMC) to a deterministic limit described by an ordinary differential equation (o.d.e.) ( see Kurtz \cite{kurtz70ode-markov-jump-processes}, Darling \cite{darling02limits-purejump-markov}), to develop deterministic (or, so called, fluid limit) approximations to the  HILT model for interest evolution.

Recall that $\mathcal{A}(k), k \geq 1$, is the set of destinations and $\mathcal{D}(k)$ the set of \emph{newly added} destinations (infectious) at the end of period $k$ with $\mathcal{D}(0) = \mathcal{A}(0)$. Define $\mathcal{B}(k) = \mathcal{A}(k-1)$ with $\mathcal{B}(0) = \emptyset$. Thus, for $k \geq 1$, $\mathcal{B}(k) = \cup_{0 \leq i \leq k-1} \mathcal{D}_i$. We will work with sets $\mathcal{B}(k)$ and $\mathcal{D}(k)$ to derive the fluid limits of the HILT process. Also, since the nodes are homogeneous in the HILT model, it suffices to record the sizes of the respective sets, not the exact membership of the sets themselves.  Let $A(k)$, $D(k)$, and $B(k)$ be the sizes of the sets $\mathcal{A}(k)$, $\mathcal{D}(k)$, and $\mathcal{B}(k)$, respectively. 
 
We can show that the original HILT process $(B(k), D(k))$ is a Markov chain (see Appendix~\ref{app:HILT-markov}). In order to obtain an approximating o.d.e., we work with an appropriately scaled Markov process $(B^{N}(t), D^{N}(t))$, which evolves on a time scale $N$ times faster than that of the original system. We can visualize this process as evolving over ``minislots'' of duration $1/N$, whereas the original process evolves at the epochs $k=0,1,2,\cdots$. The minislots are indexed by $t=0,1,2,\cdots$. Since this new process runs on a faster time scale, we scale the amount of change it undergoes in each ``minislot'' by adopting an approach taken in \cite{benaim-leboudec08mean-field-models}. In the present context, the scaling approach can be interpreted as follows. In each minislot, each infectious destination in $D^{N}(t)$ is permitted to influence the relays with probability $\frac{1}{N}$ and its influence is deferred with probability $1-\frac{1}{N}$. In the former case, it contributes its influence of $\gamma_N$ and then moves to the set $B^{N}(t+1)$, otherwise it stays in the $D^{N}(t+1)$ set. In the original process, the influence of all the newly infected destinations will be taken into account when determining the popularity level of the content in the next time step, whereas, in the scaled process, only those infectious destinations that choose to use their influence will be considered.\footnote{The scaling is done in order to obtain an analysis of the system as $N$ scales. Since the number of pairs of infectious destinations and susceptible relays scales as $N^2$, we need to further slow down the dynamics by a factor of $N$ and hence the probabilistic scaling.}. Define by $C^N(t) \subseteq D^N(t)$ the set of infectious destinations who use their influence at time $t$ and thereby become non-infectious. Then, 
\[ C^{N}(t) = \frac{D^{N}(t)}{N} + Z^{N}_{b}(t+1)\]
\begin{eqnarray*}
 B^{N}(t+1) &=& B^{N}(t) + C^{N}(t) \\
            &=& B^{N}(t) + \frac{D^{N}(t)}{N} + Z^{N}_{b}(t+1)\\
\end{eqnarray*}
 
and by applying Equation~\ref{eqn:HILT-activation} to the relay nodes $j \in \mathcal{N} \backslash \mathcal{A}(t)$, 
\begin{eqnarray*}
\lefteqn{D^{N}(t+1) = D^{N}(t) - C^{N}(t)}  \\
\hspace{-2cm} &+& \mathbb{E} \bigg[ \frac{F(\gamma_N(B^{N}(t) + C^N(t)))- F(\gamma_N B^{N}(t))}{1-F(\gamma_N B^{N}(t))} \bigg] \times \\
& & \bigg(N-B^{N}(t) -D^{N}(t)\bigg) + Z^{N}_{d}(t+1)
\end{eqnarray*}
where $Z^{N}_b(t+1)$ and $Z^{N}_d(t+1)$ are zero mean random variables conditioned on the history of the process $(B^{N}(0), D^{N}(0), \cdots B^{N}(k), D^{N}(k))$ representing the noise terms, and the expectation in the expression for $D^{N}(t+1)$ is with respect to $C^N(t)$. Define $\Tilde{B}^{N}(t) = \frac{B^{N}(t)}{N}$ and similarly $\Tilde{C}^{N}(t)$ and $\Tilde{D}^{N}(t)$ as fractions of the total population. $(\Tilde{B}^{N}(t), \Tilde{D}^{N}(t))$, $t=0,1,2,\cdots$ is a DTMC on the state space $[0, \frac{1}{N}, \frac{2}{N}, \cdots 1] \times [0, \frac{1}{N}, \frac{2}{N},\cdots 1]$ with $\Tilde{B}^{N}(t)+\Tilde{D}^{N}(t)\leq 1$. 

Then we can state the following theorem on the convergence of the scaled process to a deterministic limit. We assume that the threshold distribution $F(\cdot)$ has density $f(\cdot)$. 
\begin{theorem}
\label{thm:kurtz-theorem-hilt}
Given the scaled HILT interest evolution Markov process $( \Tilde{B}^{N}(t), \Tilde{D}^{N}(t) )$, with bounded first derivative of the density ${f}(\cdot)$ of the threshold distribution,  we have for each $T > 0$ and each $\epsilon > 0$,
\begin{eqnarray*}
\hspace{-0.4cm}\mathbb{P} \bigg( \sup_{0 \leq u \leq T} \big| \big| \big( \Tilde{B}^{N}( \lfloor Nu \rfloor ),\Tilde{D}^{N}( \lfloor Nu \rfloor ) \big) - \big( b(u),d(u) \big) \big| \big| > \epsilon \bigg)\\
\stackrel{N\rightarrow \infty}{\rightarrow} 0
\end{eqnarray*}
where $(b(u),d(u))$ is the unique solution to the o.d.e.,
\[ \dot{b} = d\]
\[ \dot{d} = \frac{f(\Gamma b)\Gamma d }{1- F(\Gamma b)} (1-b-d) - d \]
with initial condition $(b(0)=0,d(0)=a(0))$.
\end{theorem}
\begin{proof}
This is essentially an instance of Kurtz's Theorem \cite{kurtz70ode-markov-jump-processes}. See Appendix~\ref{app:hilt-fluidlimit} for the steps involved in deriving the fluid limit and a detailed proof verifying the necessary conditions for Kurtz's theorem for the HILT model. Since the o.d.e. drift equations are Lipschitz, they satisfy the Cauchy-Lipschitz condition and the system has a unique solution once the initial condition is fixed.
\end{proof}
\textit{Remark:}
The hazard function \cite{cox61renewal-theory} for the cdf $F(x)$ is given by $h_F(x) = \frac{f(x)}{1-F(x)}$ where $f(x)$ is the corresponding probability density function. Hence the o.d.e. becomes
\begin{eqnarray} 
\label{eqn:ODE_F_1}
\dot{b} &=& d  \\
\label{eqn:ODE_F_2}
 \dot{d} &=& h_F(\Gamma b) \Gamma d  (1-b-d) - d 
\end{eqnarray}
Here $h_F(\Gamma b) \Gamma d$ can be interpreted as the rate of conversion for the relays, where $(1-b-d)$ is the current fraction of relays in the population.
\subsection{Accuracy of the O.D.E.\ Approximation}
Consider the HILT process of interest evolution, under the special case of uniform distribution of influence thresholds, i.e., $\Theta_i \sim U[0,1]$. The hazard function corresponding to uniform distribution is given by, for $0 \leq x \leq 1$,  \[h_F(x) = \frac{1}{1-x}\] The corresponding o.d.e.s then become:
\begin{eqnarray}
\label{eqn:ODE_unif_1}
\dot{b} &=& d\\
\label{eqn:ODE_unif_2}
\dot{d} &=& \frac{\Gamma d}{1-\Gamma b}  (1-b-d) - d 
\end{eqnarray}
Figure~\ref{plot:hilt-convergence} illustrates the convergence of the scaled HILT process $(\Tilde{B}^N(\lfloor Nt \rfloor), \Tilde{D}^N(\lfloor Nt \rfloor))$ to the solution of the above o.d.e. $(b(t), d(t))$ for increasing values of $N$ (50,100,500,1000). $\Gamma$ and $d_0 $ were chosen to be $0.9$ and $0.2$ respectively. We see that for  $N=1000$ the o.d.e. approximates the scaled HILT Markov chain very well, as expected from Theorem~\ref{thm:kurtz-theorem-hilt}.

As we pointed out earlier, the scaled HILT process $(B^N(t), D^N(t))$ correctly captures the average dynamics of the original HILT process. Hence it can be expected that the o.d.e. solution evaluated at the ends of the original slots (indexed by $k \geq 1$), will be a good approximation to the average of the process $(B(k), D(k))$. This is illustrated in Figure~\ref{plot:kurtz-errorbar}, where multiple sample paths of the original discrete time HILT process, $(B(k), D(k))$, are superimposed on the o.d.e. solution, $(b(t), d(t))$. The good match confirms that the fluid limit obtained from the particular scaling of the Markov chain, retains the average behaviour of the original influence spread process. 

\begin{figure}
\centerline{\includegraphics[width=8cm, height=4cm]{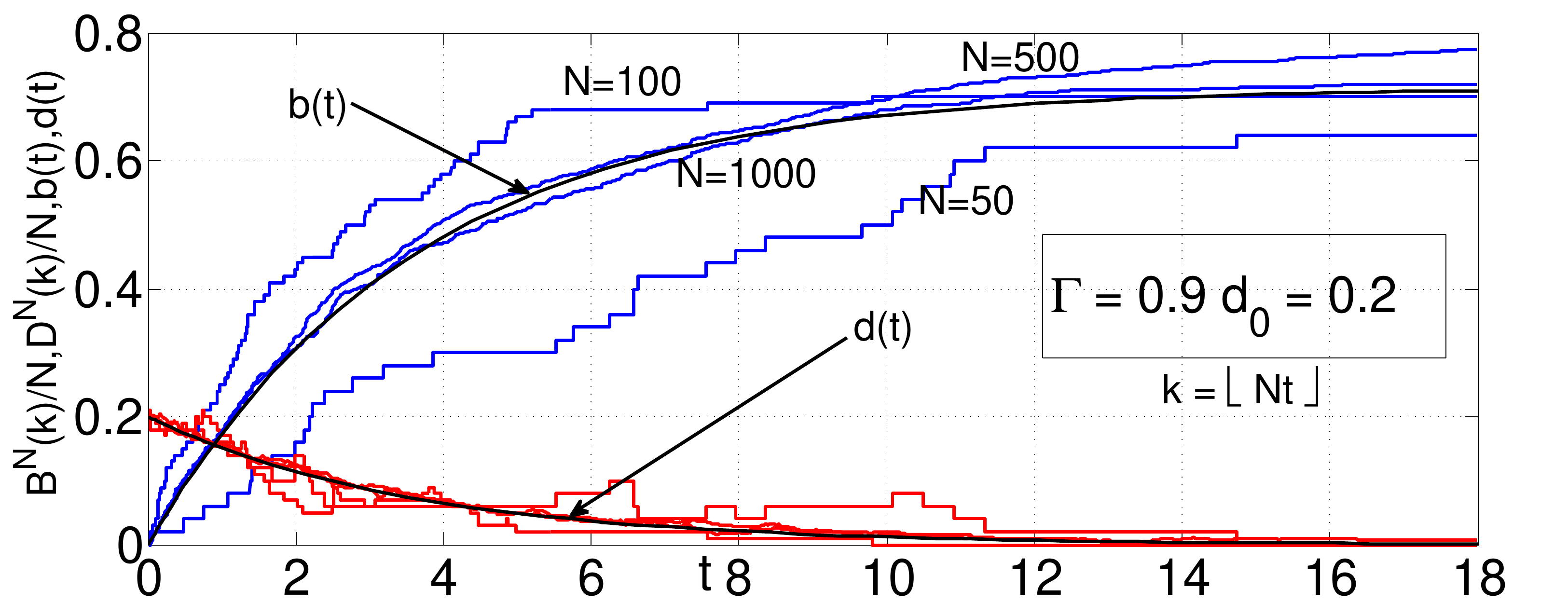}}
\vspace{-0.3cm}
\caption{Convergence of the scaled HILT Markov chain $(\Tilde{B}^{N}(k), \Tilde{D}^{N}(k))$ to the o.d.e. limit $(b(t), d(t))$. The o.d.e. solution $(b(t),d(t))$ plotted along with sample paths of the scaled process $(\Tilde{B}^{N}(k), \Tilde{D}^{N}(k))$ for $N=50,100,500,1000$.}
\label{plot:hilt-convergence}
\vspace{-0.5cm}
\end{figure}
\begin{figure}
\centerline{\includegraphics[width=8cm, height=4cm]{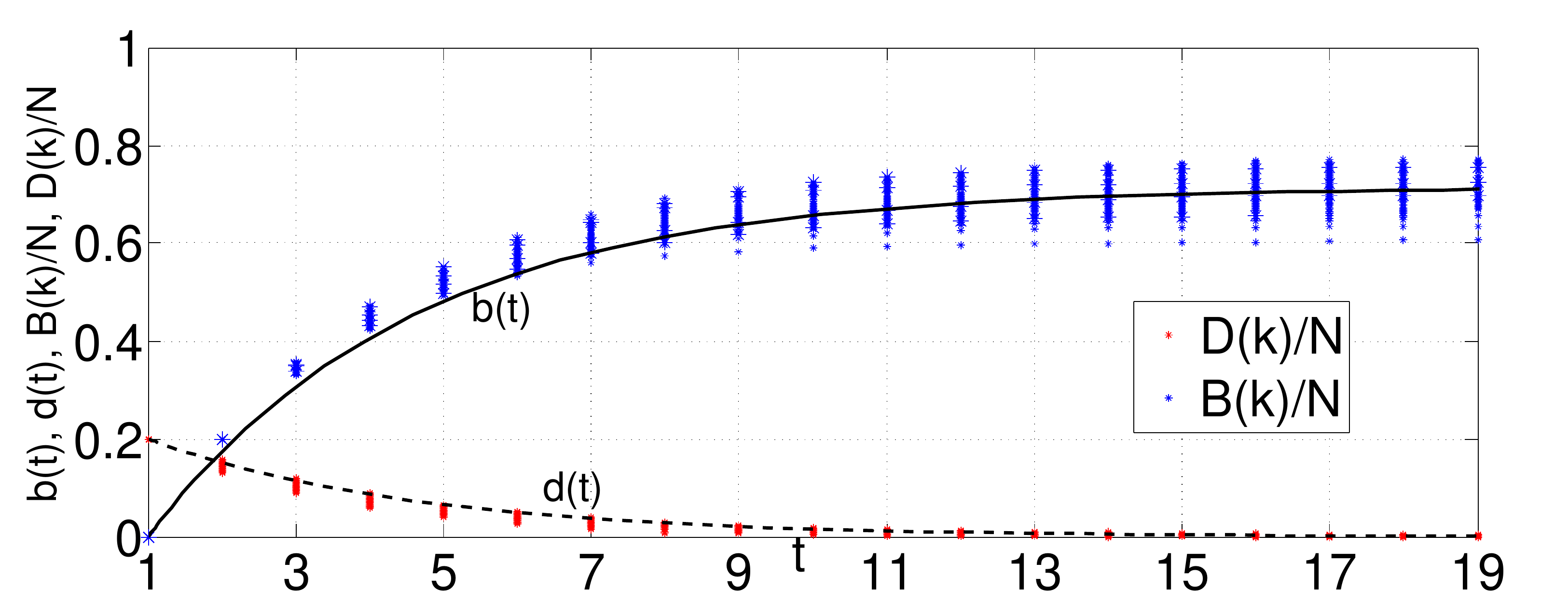}}
\vspace{-0.3cm}
\caption{Multiple sample paths of the unscaled HILT process $(B(k),D(k))$ for $N=3000$ is plotted along with the o.d.e. solution $(b(t),d(t))$ for $\Gamma = 0.9$ and $d_0=0.2$}
\label{plot:kurtz-errorbar}
\vspace{-0.5cm}
\end{figure}

\subsection{Effect of the Threshold distribution}
\label{sec:HILT-distribution}
In the HILT model, while $\Gamma$ is indicative of the total level of influence each individual can receive from the others, the threshold distribution $F(\cdot)$ captures the variation among the individuals' susceptance levels for getting interested the content. An empirical analysis on the effect of threshold distributions on collective behavior was provided in \cite{granovetter78threshold-models}. Having established the o.d.e. limit and studied its ability to approximate the evolution of interest process, we can now exploit the o.d.e. limit to study the effect of the threshold distribution.

\subsubsection{Uniform Distribution}
For the HILT model of interest evolution under the uniform threshold distribution, we can state the following theorem.

\begin{theorem}
Given the initial fraction of destinations $d_0$, in an HILT network with influence weight $\Gamma$ under the uniform threshold distribution, define $r = 1 - \Gamma + \Gamma d_0$. Then,
\begin{enumerate}
 \item The fluid limit for the evolution of interest is given by,
\begin{eqnarray}
\label{eqn:unif_explicit_1}
b(t) &=& \frac{d_0}{r} - \frac{d_0}{r} e^{-rt} \\
\label{eqn:unif_explicit_2}
d(t) &=& d_0 e^{-rt} 
\end{eqnarray}
\item The final fraction of destinations is given by 
\begin{eqnarray}
\label{eqn:d_0_b_infty}
b_\infty=\frac{d_0}{1-\Gamma+\Gamma d_0}
\end{eqnarray}
\end{enumerate}
\end{theorem}
\begin{proof}
The first part of the theorem is obtained by explicitly solving the o.d.e.s for the HILT model under uniform distribution (Equations~\ref{eqn:ODE_unif_1} and \ref{eqn:ODE_unif_2}) for the initial conditions $b(0)=0,d(0)=d_0$.
\end{proof} 

The linear threshold model under the uniform threshold distribution has been studied in discrete setting for general networks in \cite{kempe-etal03max-spread-infl}. Since in the HILT model all nodes are homogeneous, we will be interested in the influence of a set of size $k$. Consider the HILT network with $N$ nodes and influence weight $\gamma_N$. Let $I_{\gamma_N}(k)$ be the expected size of the \emph{terminal set of destinations} $\mathcal{A}(T)$, starting with $\mathcal{A}_0$ of size $k$ as the initial set of destinations. By using results from \cite{srini-kumar11LT-model-ncc}, we can show that, 
\[I_{\gamma_N}(k) = k[1 + (N-k)\gamma_N [ 1 + (N-k-1)\gamma_N[ 1 + \cdots \]

In the expression for $I_{\gamma_N}(k)$, noting $\lim_{N \rightarrow \infty} N \gamma_N = \Gamma$ and $d_0 = \frac{k}{N}$, we can show that as $N \rightarrow \infty$, $\frac{I_{\gamma_N}(k)}{N} \rightarrow b_\infty = \frac{d_0}{r}$. This reconfirms the fact that the fluid model is consistent with the discrete formulation. Also, while \cite{srini-kumar11LT-model-ncc} allows us to compute only the final fraction of destinations, our current work provides a good approximation of the actual trajectories of the influence process.

\begin{figure}[t]
\centerline{\includegraphics[scale=0.3]{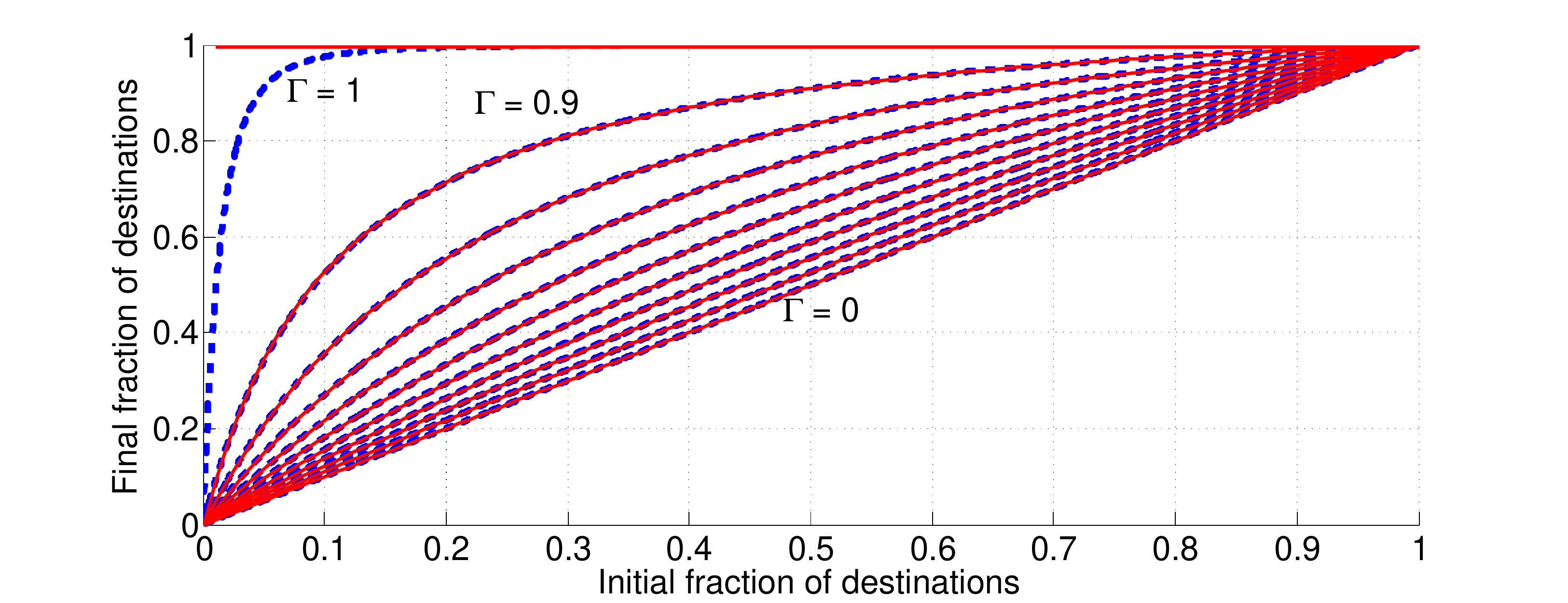}}
\vspace{-0.3cm}
\caption{Final fraction of destinations as predicted by the fluid limit $b_\infty$ (bold lines) and by the discrete formulation $\frac{I_{\gamma_N}(k)}{N}$ (dashed lines)plotted for various values of $\Gamma$ and $d_0$ and $N=3000$. The gap seen between $b_\infty$ and $\frac{I_{\gamma_N}(k)}{N}$ decreases to zero as $N \to \infty$. }
\label{plot:hilt_k_3000}
\vspace{-0.5cm}
\end{figure}

\textit{Remarks:} Figure~\ref{plot:hilt_k_3000} shows the behaviour of $b_\infty$ for various of $\Gamma$ and $d_0$. Plotted alongside are the values of $\frac{I_{\gamma_N}(k)}{N}$ for the corresponding values of $\gamma_N$ and $k$ for $N=3000$ and it can be seen that the two solutions match very well. It is easy to note that when $\Gamma=0$, the HILT process does not have any effect on the number of destinations, i.e., the fraction of destinations and relays in the population remains a constant, as, for example, in \cite{singh-etal11dtn-multi-destination}. We observe from the fluid limit that, as long as $\Gamma < 1$ we cannot influence the entire population (i.e., $b_\infty < 1$) unless we start off with the entire population to be destinations (i.e., $d_0 =1$). But if $\Gamma = 1$, then the influence of the destinations are at the maximum, and we can ultimately convert all nodes into destinations, given we start with a non-zero fraction of destinations, i.e., $b_\infty =1$ provided $d_0 > 0$. 

\subsubsection{Exponential Distribution}
\label{sec:exp-dist-hilt}
For the uniform distribution the hazard rate function is $h_F(x) = \frac{1}{1-x}$ which is monotonically increasing in $x$. Such a hazard rate implies that the population consists of relay nodes that are more susceptible to getting interested in the content as the total number of destinations increases (i.e., total accumulated influence over the past increases). The exponential distribution with parameter $\lambda$ yields $h_F(x)= \lambda$, a constant hazard rate function. This implies a \emph{memoryless} property for the influence process, i.e., the relay nodes are equally likely to get influenced at a given time instant, irrespective of the net accumulated influence in the past. We can then state the following theorem.

\begin{theorem}
 Given the initial fraction of destinations $d_0$, in an HILT network with influence weight $\Gamma$ under exponential distribution with parameter $\lambda$,
\begin{enumerate}
 \item The fluid limit of the evolution of interest is given by the solution to the o.d.e.,
\begin{eqnarray*}
 \dot{b} &=& d\\
 \dot{d} &=& \lambda \Gamma d(1-b-d) - d 
\end{eqnarray*}
 \item The final fraction of destinations is the solution to the transcendental equation, 
\[ b_\infty =  1 - (1-d_0) e^{-\lambda \Gamma b_\infty}\]
\end{enumerate}
\end{theorem}
\begin{proof}
The first part is obtained by substituting $h_F(x)= \lambda$ in the HILT o.d.e. (Equations~\ref{eqn:ODE_F_1} and \ref{eqn:ODE_F_2}). This is equivalent to the SIR epidemic model with infection rate $\lambda \Gamma$ and recovery rate of 1 \cite{kermack-mckendrick27SIR-epidemics}. Note that $b(t)$ is then equivalent to the Recovered set (R), and $d(t)$ is equivalent to the Infected set (I). The second part follows from classic SIR literature by observing that the basic reproduction number (expected number of new infections from a single infection) is $\lambda \Gamma$. 
\end{proof}

\textit{Remarks:} We see that an HILT model with exponential distribution with parameter $\lambda$, and edge weights $\gamma_N=\frac{\Gamma}{N-1}$ is equivalent, in fluid limit, to an SIR model with infection rate $\lambda \Gamma$ and recovery rate 1. Alternatively, if we consider the SIR model as modeling epidemic spread over a population of nodes, with pairwise meetings occuring at Poisson points of rate $\lambda$, then in order to get an infection rate of $\lambda \Gamma$, we need to consider a $\Gamma$-thinned version of the Poisson process, i.e. when an infectious node meets a susceptible node, it passes the infection with probability $\Gamma$. This interpretation would be crucial while considering the decentralized version of the influence spread process. 

\section{Joint Spread of Interest and Content: SIR-SI Model}
\label{sec:sir-si}
In this section, we aim to model the joint spread of interest in the content and the content itself. We do this by combining the earlier analysis on interest spread (essentially the SIR model) with a controlled-epidemic copying process for content spread (probabilistic control, similar to the Susceptible-Infective (SI) model \cite{daley-gani99epidemic-modeling}).   

Let $k \in \mathbb{N}$ index the epochs of pairwise meetings or ``recovery'' of a node (from the SIR literature). Pertaining to the content, each node has a \emph{want} state and a \emph{have} state. In order to model the content delivery process, we further classify the nodes depending on whether they have the content (i.e., based on the \emph{have} bit). Recall $\mathcal{A}(k)$, the set of destinations and let $\mathcal{S}(k) = \mathcal{N} \backslash \mathcal{A}(k)$ be the set of relays.  We further classify the destinations as either infectious ($\mathcal{D}(k)$) or non-infectious ($\mathcal{B}(k)$). Let $\mathcal{X}(k) \subseteq \mathcal{A}(k)$ denote the set of destinations that have the content, and $\mathcal{Y}(k) \subseteq \mathcal{S}(k)$ denote the set of relays that have the content. Let $\mathcal{X}_b(k)$ and $\mathcal{X}_d(k)$ respectively be the intersection of $\mathcal{X}(k)$ with the sets $\mathcal{B}(k)$ and $\mathcal{D}(k)$. 

Pairs of nodes meet at the points of a Poisson process with rate $\lambda_N = \frac{\lambda}{N}$, and a relay node gets converted into a destination, if it meets an infectious destination, and the infectious node succeeds in transmitting the influence (modeled by the influence weight $\Gamma$ resulting from the HILT model). We also include the recovery rate $\beta_N$, the rate at which infectious destinations become non-infectious. For the SIR model derived from the HILT model, $\beta_N=1$, but in this case we wish to work with a generic $\beta_N$. For the evolution of $(\mathcal{X}(k), \mathcal{Y}(k))$, we model a content copying process based on the Susceptible-Infective (SI model). Whenever a node that has the content meets a node that does not, content transfer takes place based on a controlled copying process. We always copy to a destination that does not have the content, whereas copying to a relay is controlled by a probability $\sigma \in [0,1]$.  We wish to obtain the fluid dynamics of this model.  

In the combined model, underlying the content delivery process (the SI model), the influence process (the SIR model) converts relay nodes into destinations. Thus, in our setup, the fraction of destinations is time-varying (unlike \cite{singh-etal11dtn-multi-destination}). Also, the content spread is dependent on the interest evolution but not vice versa. An interesting feature of these co-evolution models is the importance of copying to a relay node. As a content provider, we might be interested in delivering only to the destinations (interested in the content). But, there are two advantages of copying to a relay. First, copy to a relay promotes the further spread of the content even to destinations; this is the aspect explored in a controlled Markov process setting in \cite{singh-etal11dtn-multi-destination}. Second, the relay we copy to now might later get influenced to become a destination.

\subsection{System Evolution}
The set of all possible transitions among the different states of nodes is shown in Figure~\ref{fig:sir-si}. Note that the process evolves at epochs $t_k$ indexed by $k$ which are either pairwise meetings  (occuring at rate $\lambda_N |\mathcal{N}|(|\mathcal{N}|-1)$) or instances of recovery of an infectious destination (occuring at rate $\beta \mathcal{D}(k)$).  The system state is represented by the tuple $Z(k) = (B(k), D(k), X_b(k), X_d(k), Y(k))$, the sizes of the respective sets. The dependence on $N$ is implied, and not explicitly indicated in the notation. $Z(k)$ is a continuous time Markov chain; in Table \ref{table:SIRSI-rates} we show its transition structure. The type of the epoch is given by the membership of the node(s) which are involved in the pairwise meeting or recovery, and determines the state update $\delta_k$ at time $t_k$. Also note that, when a node from $\mathcal{X}_d(k)$ meets a node from $\mathcal{S}(k) \backslash \mathcal{Y}(k)$, there is a possibility of both influence spread and content copy. In such cases, we assume that the attempt to influence spread precedes content copy. If the influence succeeds (occurs with probability $\Gamma$), the relay node becomes a destination and is immediately given the content. If the influence fails, then it is treated as a relay, and the content is copied with the probability $\sigma$ (control). Thus, the state updates depend on whether a potential influence succeeded (occurs with probability $\Gamma$) and whether a potential relay node received the content (occurs with probability $\sigma$). The rate of various epochs and the corresponding state updates are listed in Table \ref{table:SIRSI-rates}. The system state does not change, for any other pairwise interaction, and hence the corresponding $\delta_k = (0,0,0,0,0)$. 

\begin{figure}
\begin{center}
\includegraphics[scale=0.35]{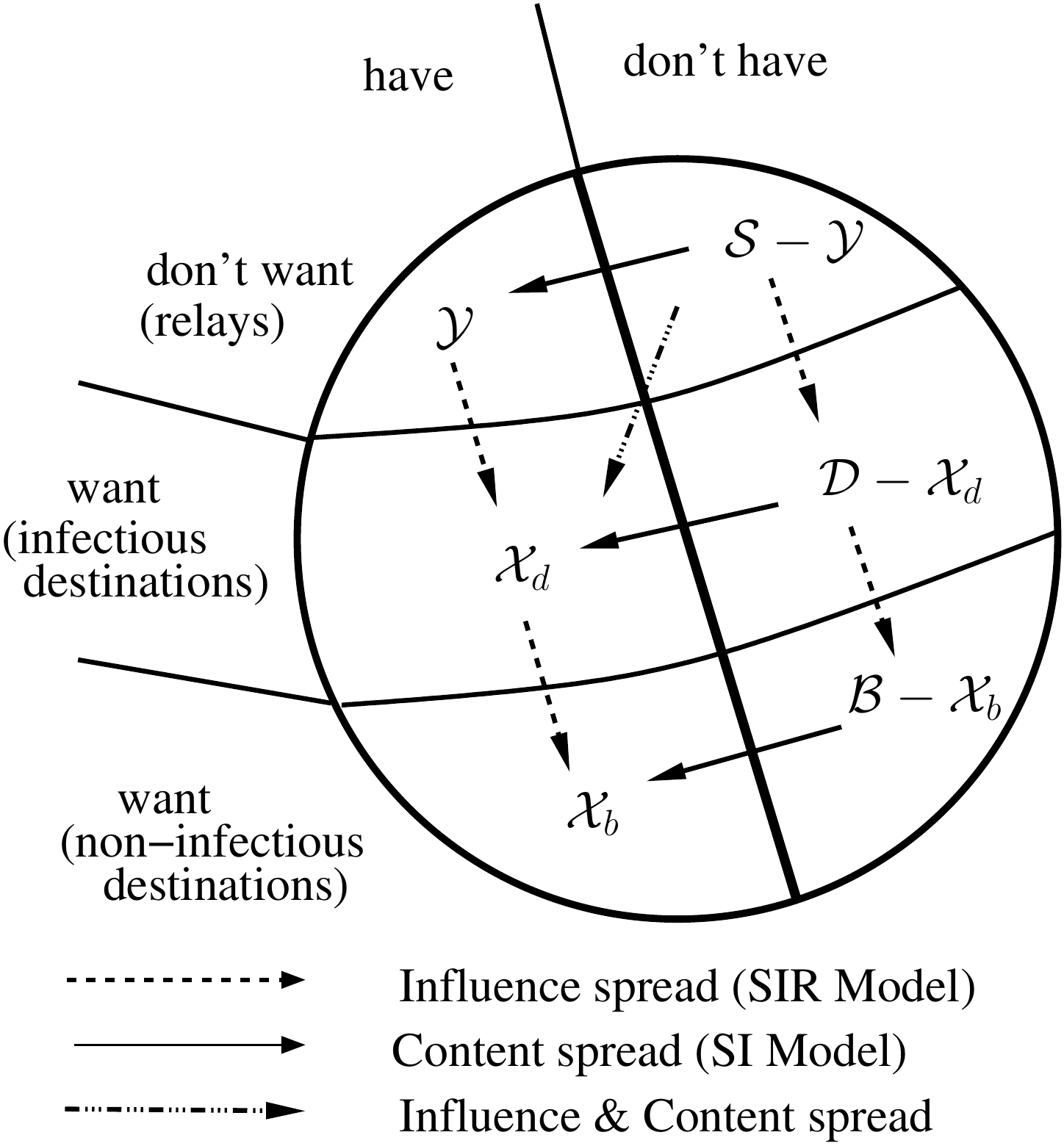} 
\end{center}
\caption{Possible transitions between the states in the SIR-SI model. Observe that, in addition to the transitions due to SIR model and SI model separately, there are instances where both can occur simultaneously. For instance, when an infectious destination node $i \in \mathcal{X}_d$ meets a relay node $j \in \mathcal{S} - \mathcal{Y}$, the relay might get converted to a destination and also receive the content, as indicated by the diagonal transition. }
\label{fig:sir-si}
\end{figure}

\begin{figure}
\begin{center}
\resizebox{8.5cm}{!}
{
\begin{tabular}{|l|l|l|l|p{5cm}}
\hline
\textbf{Epoch type $i$} & \textbf{Rate $R_k(i)$} & \textbf{State update $\delta_k(i)$}  \\
\hline
$\mathcal{D}-\mathcal{X}_d$ recovers & $\beta_N (D(k)-X_d(k))$ & (1,-1,0,0,0) \\\hline
$\mathcal{X}_d$ recovers & $\beta_N X_d(k) $ & (1,-1,1,-1,0)\\ \hline
$\mathcal{B}-\mathcal{X}_b$ meets $\mathcal{X}+\mathcal{Y}$ & $\lambda_N (B(k) - X_b(k))(X(k) + Y(k))$ & (0,0,1,0,0) \\\hline
$\mathcal{D}-\mathcal{X}_d$ meets $\mathcal{Y}$ & $\lambda_N (D(k) - X_d(k))Y(k)$ & \parbox{4cm}{(0,0,0,1,0)\\ + (0,1,0,1,-1) w.p. $\Gamma$} \\\hline
$\mathcal{D}-\mathcal{X}_d$ meets $\mathcal{X}$ & $\lambda_N (D(k) - X_d(k))X(k)$ & (0,0,0,1,0) \\\hline
$\mathcal{X}_d$ meets $\mathcal{Y}$ & $\lambda_N X_d(k)Y(k)$ & (0,1,0,1,-1) w.p $\Gamma$ \\\hline
$\mathcal{S}-\mathcal{Y}$ meets $\mathcal{X}_b + \mathcal{Y}$ & $\lambda_N (X_b(k) + Y(k))(S(k) - Y(k))$ & (0,0,0,0,1) w.p. $\sigma$\\\hline
$\mathcal{S}-\mathcal{Y}$ meets $\mathcal{D}-\mathcal{X}_d$ & $\lambda_N (D(k)-X_d(k))(S(k) - Y(k))$ & (0,1,0,0,0) w.p. $\Gamma$ \\\hline
$\mathcal{S}-\mathcal{Y}$ meets $\mathcal{X}_d$ & $\lambda_N (X_d(k))(S(k) - Y(k))$ & \parbox{4cm}{(0,1,0,1,0) w.p. $\Gamma$ \\ (0,0,0,0,1) w.p$(1-\Gamma)\sigma$}\\
\hline
\end{tabular}
}
\end{center}
\caption{Transition Rates and state updates for various possible epochs}
\label{table:SIRSI-rates} 
\end{figure}

\subsection{Drift Equations}
The evolution of system state can be given by $Z(k+1) = Z(k) + \delta_k$, where $\delta_k$ depends upon the type of the epoch. We can express the expected drift of the system by appropriately using the rates of various epochs. Define $\Tilde{Z}(k) = Z(k)/N$ and define the mean drift rate as follows:
\[ F^N_z(k) = \sum_{i \in \mathcal{E}} R_k(i) \delta_k(i) \]
where $i \in \mathcal{E}$ indexes the epoch type, $R_k(i)$ and $\delta_k(i)$ as given in Table~\ref{table:SIRSI-rates}. See Appendix~\ref{app:kurtz-applied-sirsi} for the exact expressions for the mean drift rates.

Consider the o.d.e.s given below. 
\begin{eqnarray}
\label{eqn:system-begin}
\dot{b} &=& \beta d\\
\dot{d} &=& - \beta d + \lambda \Gamma d s \\
\label{eqn:system-mid}
\dot{x_b} &=& \beta x_d + \lambda (b - x_b)(x+y) \\
\dot{x_d} &=& \Gamma \lambda (d - x_d)y + \Gamma \lambda x_d s + \\ \nonumber
&& \hspace{2cm} \lambda (d-x_d)(x+y) - \beta x_d \\
\dot{y} &=& -\Gamma \lambda d y + \lambda \sigma (s-y)(x_b +y +(1-\Gamma)x_d) 
\label{eqn:system-end}
\end{eqnarray}
We can then state the following result.

\begin{theorem}
\label{thm:kurtz-theorem-sirsi}
Given the coevolution Markov process $ \Tilde{Z}^N(t) = (\Tilde{B}^{N}(t), \Tilde{D}^{N}(t),\Tilde{X}_b^N(t), \Tilde{X}_d^N(t), \Tilde{Y}^N(t))$ we have for each $T > 0$ and each $\epsilon > 0$,
\begin{eqnarray*}
\mathbb{P} \bigg( \sup_{0 \leq u \leq T}  \big| \big| \Tilde{Z}^{N}(\lfloor Nu \rfloor) - z(u) \big| \big| > \epsilon \bigg) \stackrel{N\rightarrow \infty}{\rightarrow} 0
\end{eqnarray*}
where $z(u) = (b(u),d(u), x_b(u), x_d(u), y(u))$ is the unique solution to the o.d.e. system given by Equations~(\ref{eqn:system-begin})-(\ref{eqn:system-end}) with $z(0) = (0,d(0),0,x_d(0),0)$.
\end{theorem}

\begin{proof}
 The proof involves verifying the conditions for applying Kurtz's theorem to the SIR-SI process (see Appendix~\ref{app:kurtz-applied-sirsi}). Since the drift equations are Lipschitz, the uniqueness of the solution is guaranteed once the initial condition is fixed, by the Cauchy-Lipschitz condition.
\end{proof}

\subsection{Accuracy of the O.D.E.\ Approximation}
\label{accuracy-sirsi}
Figures~\ref{plot:sirsi-convergence-sigma100} and \ref{plot:sirsi-convergence-sigma30} illustrate the convergence of the scaled coevolution process $\Tilde{Z}^N(t)$ to the o.d.e. solutions $z(t)$ for increasing values of $N$  for $\sigma =1$ and $\sigma = 0.3$. We plot $a(t) = b(t) + d(t)$, $x(t) = x_b(t) + x_d(t)$ and $y(t)$ for clarity. Observe that as $N$ increases, the approximation of the original process by the o.d.e. becomes better. Observe that when $\sigma$ is increased from $0.3$ to $1$, there is significant increase in $y(t)$ but the contribution to $x(t)$ is not significant. Also, for static control, $x(t)$ asymptotically reaches $a(t)$ irrespective of $\sigma$ and $y(t)$ asymptotically reaches $s(t)=1-a(t)$ provided $\sigma >0$. Note that from the o.d.e. equations, it is clear that, $x(t)$ is monotonically increasing, whereas in general $y(t)$ need not be monotonic. 

\begin{figure}
\centerline{\includegraphics[width=9cm, height=4cm]{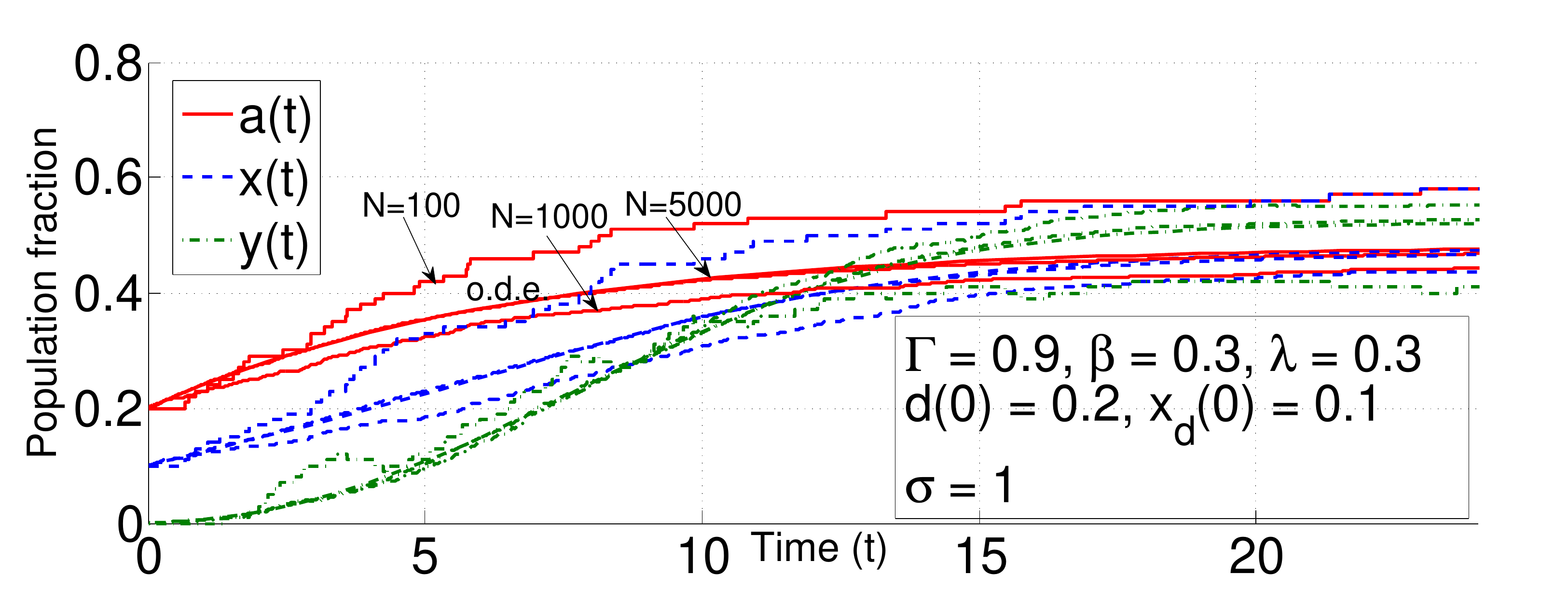}}
\vspace{-0.3cm}
\caption{Convergence of the scaled coevolution process $\Tilde{Z}^N(t)$ to the o.d.e. solutions $z(t)$ for increasing values of $N$ with $\sigma=1$}
\label{plot:sirsi-convergence-sigma100}
\vspace{-0.5cm}
\end{figure}

\begin{figure}
\centerline{\includegraphics[width=9cm, height=4cm]{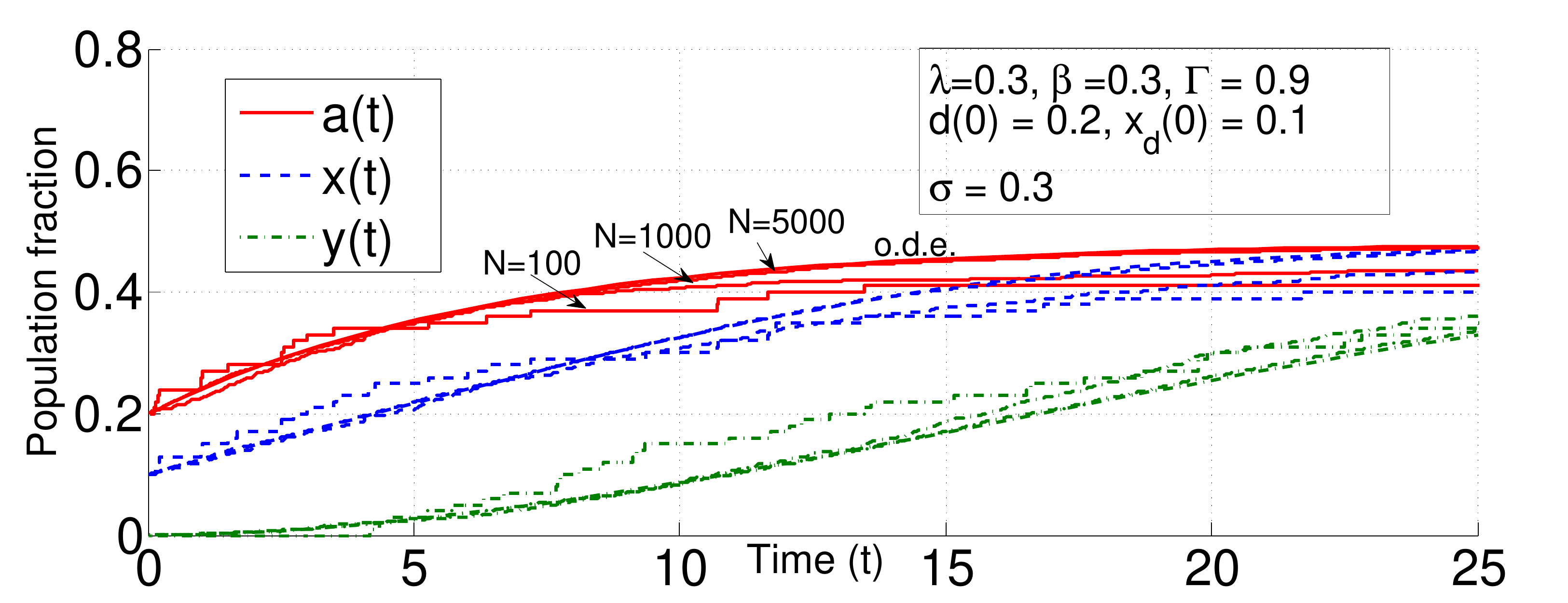}}
\vspace{-0.3cm}
\caption{Convergence of the scaled coevolution process $\Tilde{Z}^N(t)$ to the o.d.e. solutions $z(t)$ for increasing values of $N$ with $\sigma=0.3$}
\label{plot:sirsi-convergence-sigma30}
\vspace{-0.5cm}
\end{figure}

\subsection{Copying to Relays: Static vs. Dynamic Control}
In the above section we considered $\sigma$ to be a static control, i.e. $\sigma(t) = \sigma$, $\forall t$. We could instead consider a dynamic copying control $\sigma(t) \in [0,1]$.  Then for each $N$ we have a controlled continuous time Markov chain. The optimal control appears difficult to obtain in the finite size problem. We instead replace the constant copying probability in the ODE limit by $\sigma(t)$ which yields a controlled ODE. We can then obtain the optimal (deterministic) control for the controlled ODE (as in \ref{sec:optimal-control}). In certain situations it can be shown that this control is asymptotically optimal for the finite size problem as $N$ increases \cite{gast10mean-field-mdp}. The proof given in \cite{gast10mean-field-mdp}, for discrete time MDPs with finite time horizon, does not directly apply here, and needs to be extended to accommodate our setting. We do not prove this, and henceforth treat $\sigma(t)$ as a heuristic, and obtain the optimal deterministic control. 

In order to justify this assumption, the convergence of the co-evolution process to its o.d.e. limit is numerically demonstrated in Figure~\ref{plot:sirsi-convergence-thresh4} for a time-threshold type control. In this case,  $\sigma(t) = 1$, $t < 4$ and $\sigma(t) = 0$ for $t \geq 4$. Note that when $\sigma(t)=0$, $\dot{y} \leq 0$ from Equation~(\ref{eqn:system-end}). Observe that the trajectory of $x(t)$ is very similar to the one obtained when $\sigma = 0.3$, but the value of $y(t)$ for large $t$ is considerably lesser. This indicates that this threshold policy ($\tau = 4$) will perform better than the static policy $\sigma=0.3$, in the context of the objective function to be posed in the next section.

\begin{figure}
\centerline{\includegraphics[width=9cm, height=4cm]{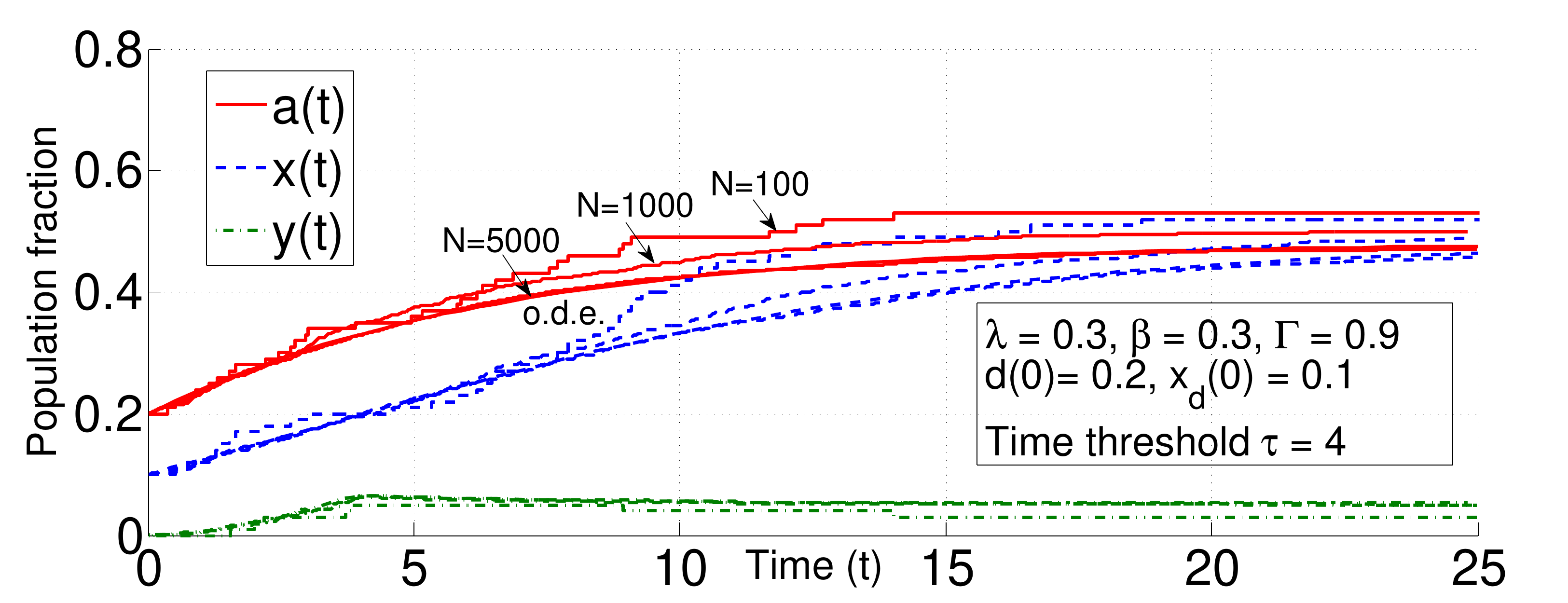}}
\vspace{-0.3cm}
\caption{Convergence of the scaled co-evolution process $\Tilde{Z}^N(t)$ to the o.d.e. solutions $z(t)$ for increasing values of $N$ with dynamic  $\sigma(t)$, a time-threshold function in this case, $\sigma(t) = 1$, $t < 4$ and $\sigma(t) = 0$ for $t \geq 4$.}
\label{plot:sirsi-convergence-thresh4}
\vspace{-0.5cm}
\end{figure}

\subsection{The Optimization Problem}
\label{subsec:optim-prob}
In \cite{singh-etal11dtn-multi-destination}, since the fraction of destinations was constant, it was suitable to choose the time of delivery to a fraction $\alpha$ of destinations as the objective to be minimized. In our setup, since the fraction of destinations is time-varying, we define the \emph{target time},

\[ T_\sigma = \inf \{t: x_\sigma(t) \geq \alpha a(\infty)\} \]
where $a(\infty)$ is the terminal fraction of destinations as given by the SIR model for interest evolution $(b(t), d(t))$ and $x_\sigma(t)$ is the fraction of destinations that have the content at time $t$. Note that since $d(\infty)=0$ and $a(t)=b(t)+d(t)$, $a(\infty) = b(\infty)$.

The cost we are interested in optimizing is of the form 
\begin{eqnarray}
\label{eqn:obj-function}
C_\sigma =  \psi y_\sigma(T_\sigma) + T_\sigma = \psi y_\sigma(T_\sigma) + \int_{0}^{T_\sigma} 1 dt
\end{eqnarray}
Here $y_\sigma(T_\sigma)$ denotes the number of relays that have the content at time $T_\sigma$ and signifies the number of wasted copies, and $\psi$ is the tradeoff parameter. 

Even though the copying cost is distributed across the nodes, we treat the number of wasted copies at the target time $T_\sigma$ as part of the system objective. This can be motivated by considering an incentive mechanism for the relay nodes which still hold the content but are not converted into destinations by the target time $T_\sigma$. Consider an example, where the content is a discount coupon for some service/product. The service provider, initially provides these discount coupons to a subset of its customers, who are then encouraged to spread replicas of the coupon. Nodes that are interested in the service, and are in possession of the coupon immediately claim the service. The service provider stops accepting discount coupons once a pre-defined number (here $\alpha a(\infty)$) of coupons have been used. At this time, the relay nodes that are still in possession of the coupon will need to be provided a reimbursement for helping the spread of the coupon. 

\subsection{Optimal Control}
\label{sec:optimal-control}
In this section we will establish the optimality of a time-threshold control for the objective given by the Equation~(\ref{eqn:obj-function}). We will use definitions and lemmas provided in in order to prove the following theorem.  

\begin{theorem}
\label{thm:optimal-control}
For the o.d.e system given by Equations~(\ref{eqn:system-begin})-(\ref{eqn:system-end}) there exists an optimal control of the form, 
\begin{eqnarray}
\label{eqn:optimal-control}
 \sigma_\tau(t) = \left\{\begin{array}{cl} 
                              1, & 0 < t < \tau \\
			      0, & t \geq \tau
			      \end{array}\right.  
\end{eqnarray}
which optimizes the cost in Equation~(\ref{eqn:obj-function}).
\end{theorem}
\begin{proof}
Recall $y_\sigma(T_\sigma)$ is the amount of wasted copies, at the target time $T_\sigma$. When $\sigma(t) = \sigma_{\tau}(t)$, a time-threshold policy (Equation~(\ref{eqn:optimal-control})), we will denote this by $y_\tau(T_\tau)$, where $\tau$ is the time threshold associated with $\sigma_\tau(t)$.

The sketch of the proof is as follows. We first establish that, the set of values taken by $y_\tau(T_\tau)$ form an interval of $[0,\rho_{\max}]$. Then, given any policy $\sigma(t)$, we show that:
\begin{itemize}
 \item If $y_\sigma(T_\sigma) = \rho \leq \rho_{\max}$, then $\exists$ a time-threshold policy $\sigma_\tau(t)$ such that  $y_\tau(T_\tau)=\rho$ and $T_\tau \leq T_\sigma$.  
 \item If $y_\sigma(T_\sigma) = \rho > \rho_{\max}$, we can find a time-threshold policy $\sigma_\tau(t)$ which has a smaller total cost. 
\end{itemize}
Thus in either case, we have a time-threshold policy, which performs at least as good as the given policy, which proves the optimality of the time-threshold policy. See Lemmas~\ref{lemma:tau-rho-continuity}, \ref{lemma:threshold-dominates-samerho} and \ref{lemma:rho-greater-than-rhomax} of Appendix~\ref{app:coop-system} for proofs of the above claims. 

\end{proof}

Figure~\ref{plot:optimal_control_tau_sweep} shows the variation of $T_\tau$, $y_\tau(T_\tau)$ and $C_\tau$ as a function of $\tau$ for fixed system parameters. It can be seen that as $\tau$ is increased, $T_\tau$ decreases monotonically, and we see an increase in the value of $y_\tau(T_\tau)$. The optimal threshold $\tau^\star$ minimizes the total cost $C_\tau$, by balancing the two component costs, taking the tradeoff parameter $\psi$ into account.

\begin{figure}
\centerline{\includegraphics[width = 8cm, height = 4cm]{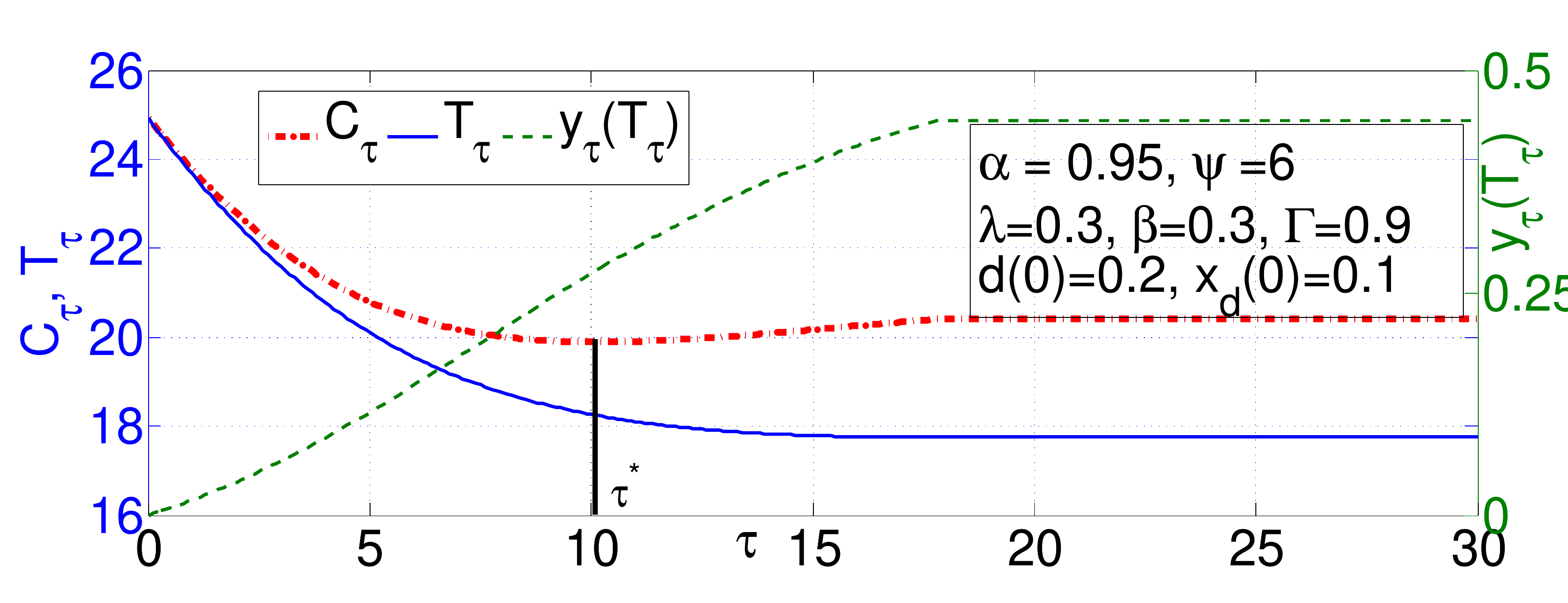}}
\caption{Evaluation of optimal $\tau$ for the time-threshold policy. The plot shows the behaviour of $T_\tau$ and $y_\tau(T_\tau)$ for various values of $\tau$, for fixed system parameters.}
\label{plot:optimal_control_tau_sweep}
\vspace{-0.5cm}
\end{figure}

\section{Numerical results}
\label{sec:numerical}
In Sections \ref{sec:HILT-model} and \ref{sec:sir-si} we have studied Markov models for influence spread and for the co-evolution of interest in an item of content and the spread of the content in a mobile opportunistic network. The spread of content can be controlled by the probability of copying the content to uninterested nodes. We established fluid limits for these Markov models and illustrated their accuracy via simulations. In this section, we provide an extensive parametric numerical study of the fluid models, and report insights that can be obtained. While the first part of the section deals with optimizing for the interest evolution process considered in isolation, the latter part deals with optimizations for the co-evolution process. We will study the interest evolution under the uniform threshold distribution, since the SIR model (which arises out of exponential threshold distribution) is well studied in the literature. 

\subsection{Interest Evolution}
Content creators are often interested in understanding the evolution of popularity of their content, and would wish to maximize the level of popularity achieved. This, in our model, is equivalent to the final fraction of destinations (nodes that are interested in receiving the content). In most cases, the content creator does not have control over the influence weight $\Gamma$ or the threshold distribution $F(\cdot)$ of the population. In order to increase the spread of interest in the content, the only parameter that can be controlled is $d_0$, the initial fraction of destinations in the population. In Section~\ref{sec:HILT-fluid}, we derived the relation between $d_0$ and $b_\infty$, the final fraction of destinations. We might be interested in choosing the right $d_0$ which can give us the required $b_\infty$, and we see that by rearranging Equation~(\ref{eqn:d_0_b_infty}) we get,
\begin{eqnarray}
\label{eqn:b_infty_d_0}
 d_0 = \frac{b_\infty (1 -\Gamma)}{1 - b_\infty \Gamma} 
\end{eqnarray}
As discussed in Section~\ref{sec:HILT-fluid}, the above equation is applicable for $\Gamma<1$. When $\Gamma=1$, from Equation~(\ref{eqn:d_0_b_infty}) we see that $b_\infty=1$ as long as we start with $d_0 >0$. 

Since the o.d.e. provides a good approximation for the temporal evolution of interest, we can also obtain results for the time taken by the process for the spread of influence.
\begin{theorem}
Given the initial fraction of destinations $d_0$ in an HILT network with parameter $\Gamma$, the time we have to wait to get the final fraction of destinations to be at least $\beta$ ($\beta < \frac{d_0}{r}$) is given by,
\[ T(\beta, d_0, \Gamma) = \frac{1}{r} \ln \bigg( \frac{1-r}{1-\frac{\beta}{d_0}r} \bigg) \]
where $r = 1 - \Gamma + \Gamma d_0$.
\end{theorem}

\begin{proof}
Firstly, note that since $a_\infty = b_\infty = \frac{d_0}{r}$, we cannot reach $\beta > \frac{d_0}{r}$. Since we are observing the process at a finite time $T$, $d(T)$ is not zero. Hence, we consider $a(T)=b(T)+d(T)$ and set it to $\beta$. We get,
\begin{eqnarray}
\label{eqn:hilt_d_0_finite_T}
a(T) = d_0 \big( \frac{1}{r} - ( \frac{1}{r}-1) e^{-rT} \big) = \beta
\end{eqnarray}
Rearranging terms,we get the expression for $T(\beta, d_0, \Gamma)$. 
\end{proof}
\begin{figure}
\centerline{\includegraphics[width = 9cm, height = 4cm]{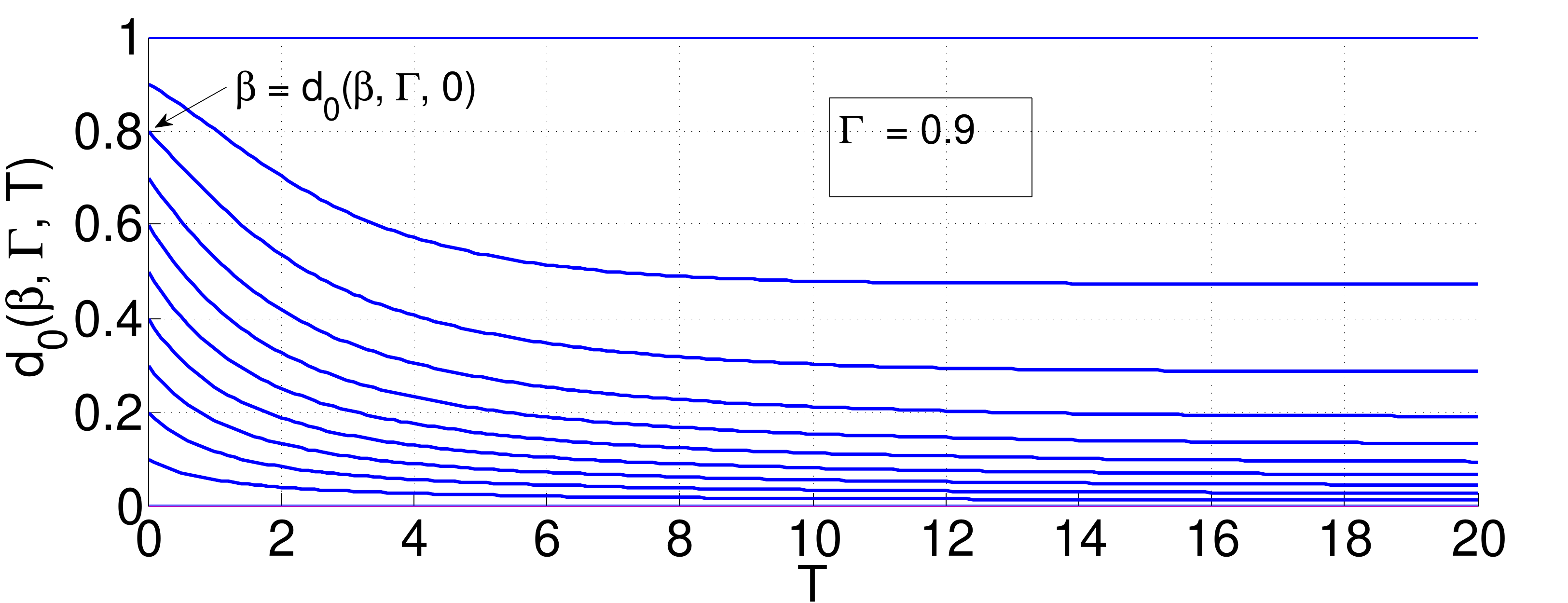}}
\vspace{-0.3cm}
\caption{Variation of $d_0^{\star}$ versus $T$ for various values of the required target $\beta$ with $\Gamma=0.9$. Notice that the value of $\beta$ for each curve is the intercept of the curve with the vertical axis }
\label{plot:d_0_fixed_Gamma_varying_beta_versus_T}
\vspace{-0.5cm}
\end{figure}

\begin{figure}
\centerline{\includegraphics[width = 9cm, height = 4cm]{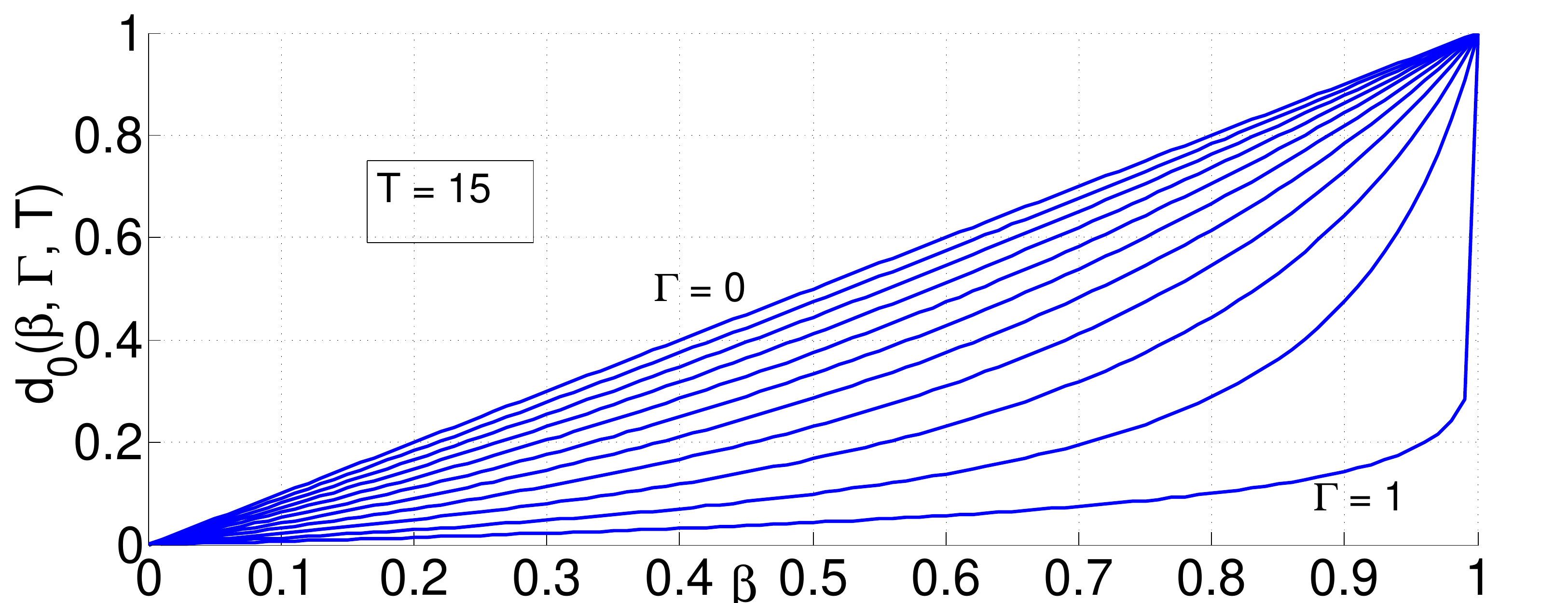}}
\vspace{-0.3cm}
\caption{Variation of $d_0^{\star}$ versus the required target $\beta$ for various values of $\Gamma$ at $T=15$}
\label{plot:d_0_fixed_T_varying_Gamma_versus_beta}
\vspace{-0.5cm}
\end{figure}

Another interesting question would be to determine the initial fraction $d_0$ to be chosen so that by time $T$ we will have at least a fraction $\beta$ of the nodes in the destination set, in the HILT network with parameter $\Gamma$. This can be solved numerically using the following equation obtained from Equation~(\ref{eqn:hilt_d_0_finite_T}).
\[ e^{-rT} = \frac{1-\frac{\beta}{d_0} r}{1-r} \]
Substituting for $r=1-\Gamma - \Gamma d_0$, 
\[ e^{-(1-\Gamma)T} e^{-\Gamma d_0} = \frac{1-\frac{\beta (1-\Gamma)}{d_0} - \Gamma \beta}{\Gamma (1-d_0)}\]

Let $H(d_0)$ and $G(d_0)$ denote the LHS and RHS of the above equation. We know that $d_0^{\star}$ that satisfies the above equality will lie in $[\frac{\beta (1 -\Gamma)}{1 - \beta \Gamma}, 1]$ and that the solution is unique, since $a(T)$ is a monotonic function in $d_0$. Also observe that $H(d_0)$ is decreasing in $d_0$ and $G(d_0)$ is increasing in $d_0$. Hence for $d_0 < d_0^{\star}$, $H(d_0) < G(d_0)$ and for $d_0 > d_0^{\star}$, $H(d_0) > G(d_0)$. Under the above conditions, we can use the iterative bisection method that will converge to $d_0^{\star}$. 

The variation of $d_0^{\star}$ with respect to the parameters $\beta$, $\Gamma$ and $T$ can be seen in Figures~\ref{plot:d_0_fixed_Gamma_varying_beta_versus_T} and \ref{plot:d_0_fixed_T_varying_Gamma_versus_beta}. As expected, for fixed $\Gamma$ and $T$, a higher value of $\beta$ requires a higher $d_0$ (Figure~\ref{plot:d_0_fixed_T_varying_Gamma_versus_beta}) and if the network has a high $\Gamma$, then it is sufficient to start with a smaller $d_0$ to reach a given $\beta$ by $T$ (Figure~\ref{plot:d_0_fixed_T_varying_Gamma_versus_beta}). Observe that, as $T$ is increased (the time constraint is relaxed), the required initial $d_0$ decreases, from $\beta$ for $T=0$ and asymptotically approaches $d_0 = \frac{\beta(1-\Gamma)}{1-\Gamma \beta}$ (obtained by setting $b_\infty = \beta$ in Equation~(\ref{eqn:b_infty_d_0})) (Figure~\ref{plot:d_0_fixed_Gamma_varying_beta_versus_T}).  

\subsection{Decentralized Influence Spread}
Having observed that the optimal control is of the form \ref{eqn:optimal-control}, we can now numerically compute the optimal time threshold, which we shall refer to as $\tau^\star$. 

Recall the cost function is of the form
\begin{eqnarray*}
C_\sigma = T_\sigma + \psi y_\sigma(T_\sigma) = \psi y_\sigma(T_\sigma) + \int_{0}^{T_\sigma} 1 dt
\end{eqnarray*}
where $T_\sigma$ is the target time to reach a fraction $\alpha$ of the terminal fraction of destinations under the policy $\sigma(t)$,i.e.,
\[ T_\sigma = \inf \{t: x_\sigma(t) \geq \alpha a_\infty\} \]
Since in this case, $\sigma(t) = \sigma_{\tau^\star}(t)$, we shall denote the total cost by $C_{\tau^\star}$, and the two components of the cost by $T_{\tau^\star}$ and $y_{\tau^\star}(T_{\tau^\star})$. In the following discussion we shall numerically study the effect of various cost parameters, system parameters and initial conditions on $\tau^\star$,$C_{\tau^\star}$, $T_{\tau^\star}$ and $y_{\tau^\star}(T_{\tau^\star})$.

\subsubsection{Effect of initial conditions}
Figure~(\ref{fig:d_0_sweep}) shows the effect of $d_0$ on the optimal $\tau^\star$. We see that as we increase $d_0$ keeping the fraction of $\frac{x_{d_0}}{d_0}$ constant, there are more destinations that have the content. Further, there are fewer relays in the population, and with fixed $\Gamma$ and $\beta$, there is less chance for them to get converted into destinations. This prevents us copying to more relays and hence there is a decrease in the value of $\tau^\star$. Note that when $d_0 = 1$, the entire population consists of only destinations, and hence the optimal $\tau = 0$.

Figure~(\ref{fig:x_d_0_sweep}) shows the effect of $x_{d_0}$ on the optimal $\tau^\star$. As the fraction $\frac{x_{d_0}}{d_0}$ is increased, since a larger number of destinations have the content, the newly infected relays also obtain the content. Observe that $\frac{x_{d_0}}{d_0} = 1$ implies that all the initial destinations are given the content. This implies that any destination converted in the future, will automatically receive the content. This alleviates the need to copy to any of the relays, since the main purpose of copying to relays is to serve the destinations without the content (either because they were destinations not given the content initially or were converted by other destinations without the content at a later time). 

\begin{figure}[t]
 \includegraphics[width = 8cm, height = 4cm]{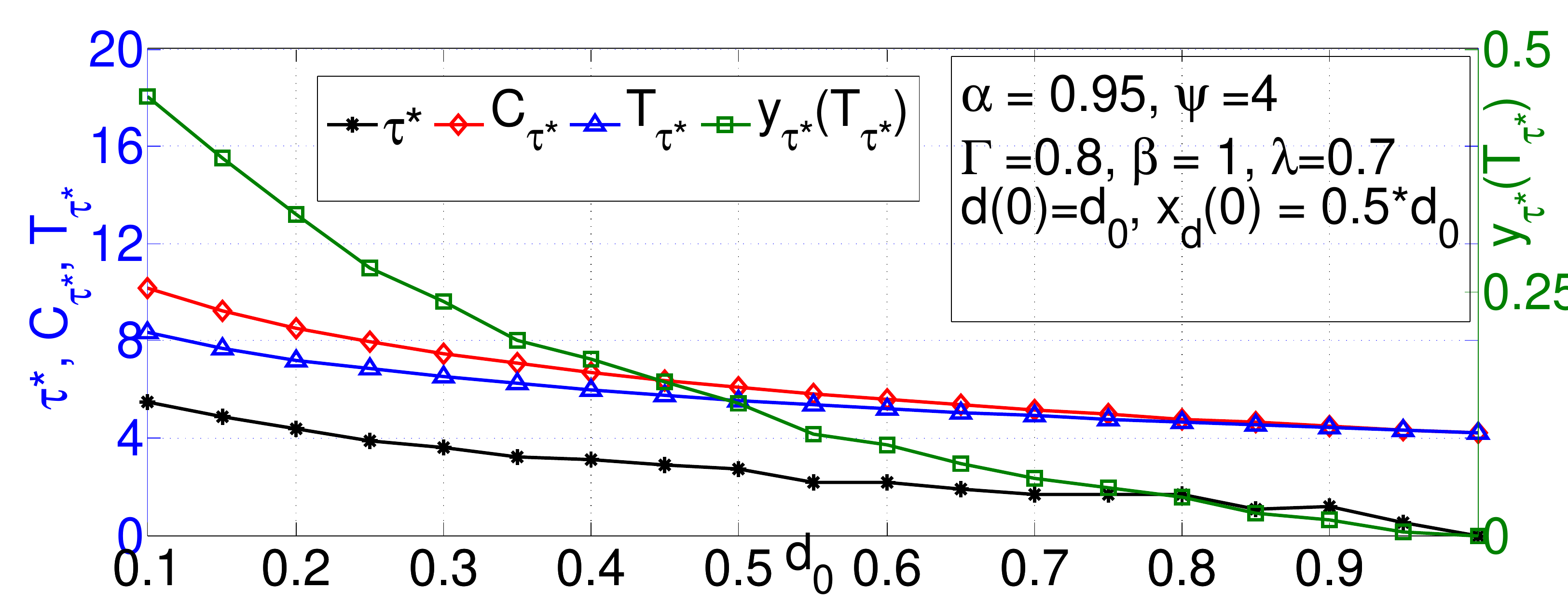}

 \caption{Decentralized Influence Spread: Effect of $d_0$ on the optimal cost and optimal $\tau$}
 \label{fig:d_0_sweep}
\end{figure}
\begin{figure}[t]
 \includegraphics[width = 8cm, height = 4cm]{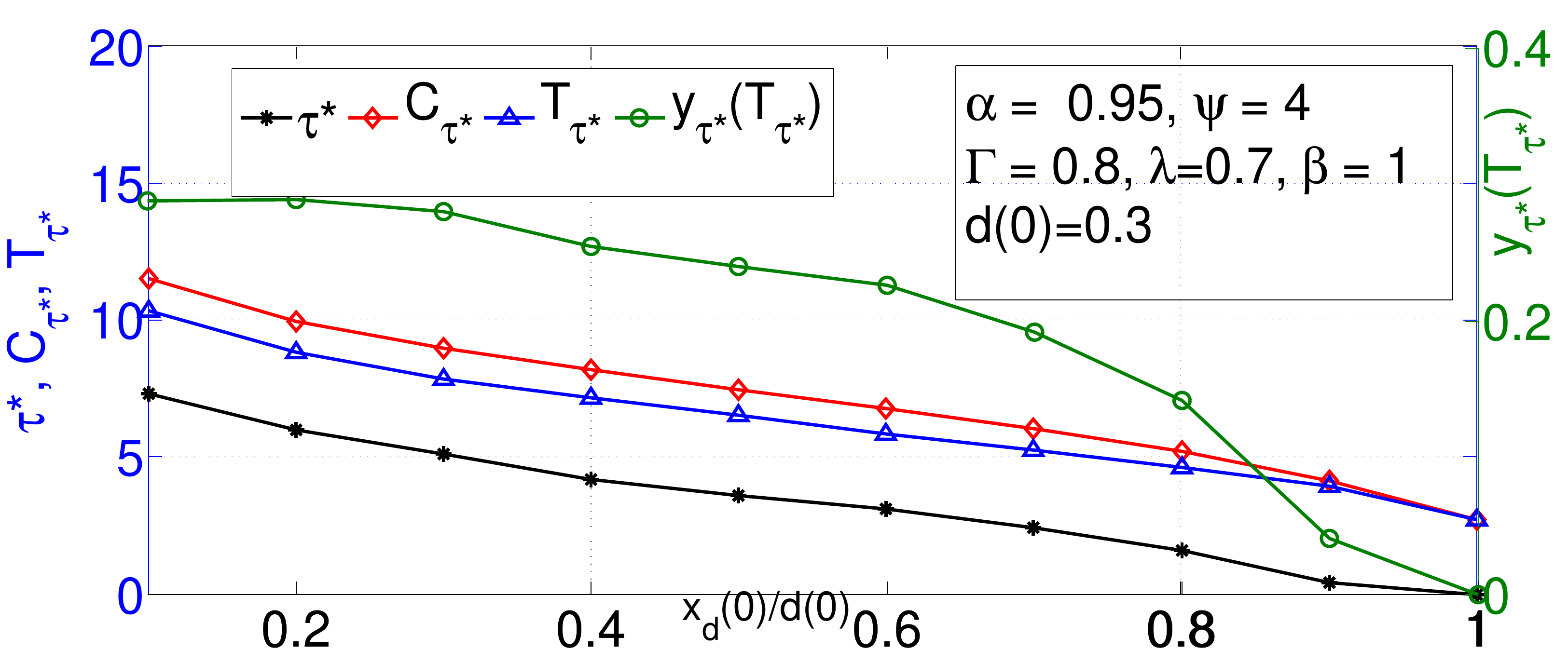}
 
 \caption{Decentralized Influence Spread: Effect of $\frac{x_{d_0}}{d_0}$ on the optimal cost and optimal $\tau$}
\label{fig:x_d_0_sweep}
\end{figure}

\subsubsection{Effect of system parameters}
Figure~(\ref{fig:gamma_sweep}) shows the effect of $\Gamma$ on $\tau^\star$.  We see that as $\Gamma$ increases, it increases the rate at which relays get converted into destinations, without affecting the rate at which content copying occurs (dependent on the meeting rate $\lambda$ and copying probability $\sigma(t)$). We do not see much effect on $\tau^\star$, and $T_{\tau^\star}$. This is because, though increasing $\Gamma$ increases the terminal set of destinations (and thus the target $\alpha a(\infty)$), it also helps in achieving the target by the conversion of relays with the contents into destinations with the content. Also due to this conversion, we observe a decrease in $y_\tau(T_\tau)$, further indicated by the fact that $y(t) \leq s(t)$ and the fraction of relays $s(t)$ is decreasing rapidly.

Figure~(\ref{fig:beta_sweep}) shows the effect of $\beta$ on the optimal $\tau^\star$. We observe that when $\beta$ is very low, destinations remain infectious for a longer duration before recovery. This leads to a faster decay in $s(t)$ the size of relays and also an increase in the target $\alpha a(\infty)$, and hence allowing/requiring a more aggressive copying policy (larger $\tau$). Also for very large $\beta$, the total fraction of destinations carrying the content would be less, and hence to make the content more available, we need to copy to more relays. Also, since in this case, the rate of conversion to destinations is much lower, we note that a relay with the content has much lesser chance of getting converted into a destination. This explains the increase in $y_\sigma(T_\sigma)$. 

Figure~(\ref{fig:lambda_sweep}) shows the effect of $\lambda$ on the optimal $\tau^\star$. As $\lambda$ increases, the rate at which content spread increases, and thus a destination in need of the content is highly likely to obtain it from another destination. This results in a more passive policy of copying to relays, leading to a decrease in the optimal $\tau^\star$ and the corresponding costs.

\begin{figure}[t]
 \includegraphics[width = 8cm, height = 4cm]{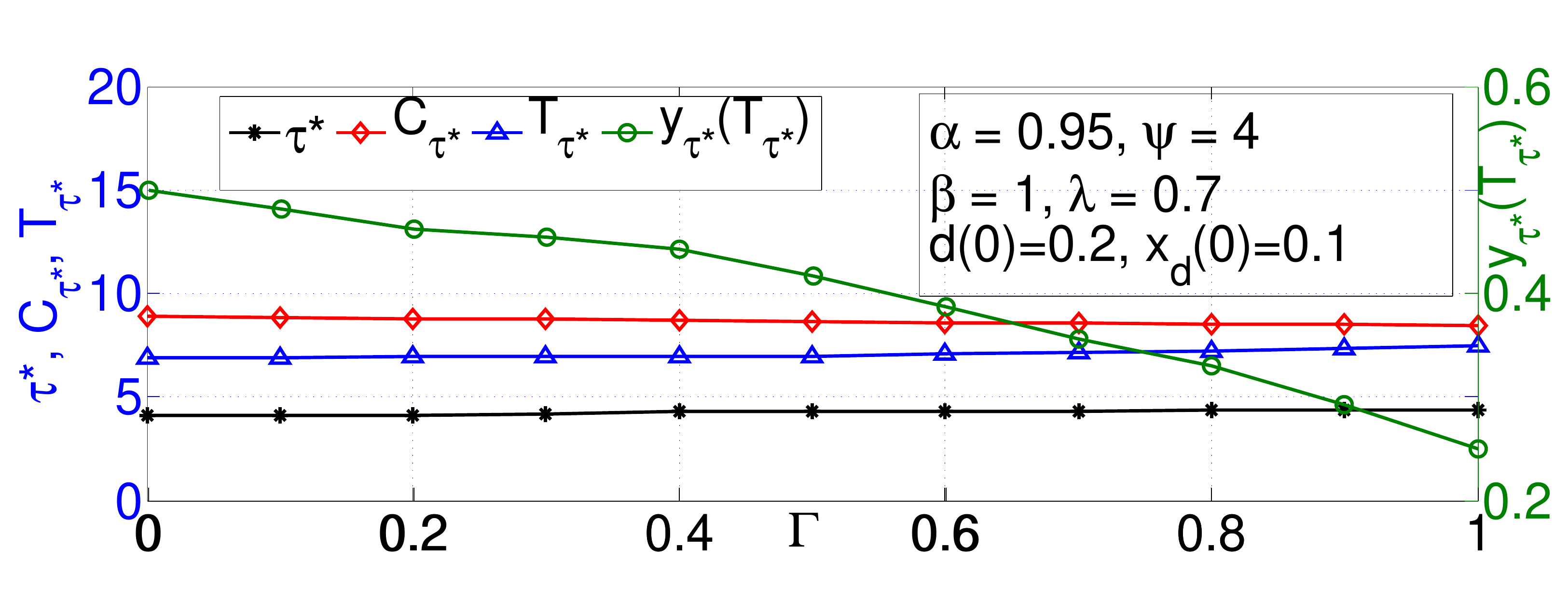}
 
 \caption{Decentralized Influence Spread: Effect of $\Gamma$ on the optimal cost and optimal $\tau$}
\label{fig:gamma_sweep}
\end{figure}

\begin{figure}[t]
 \includegraphics[width = 8cm, height = 4cm]{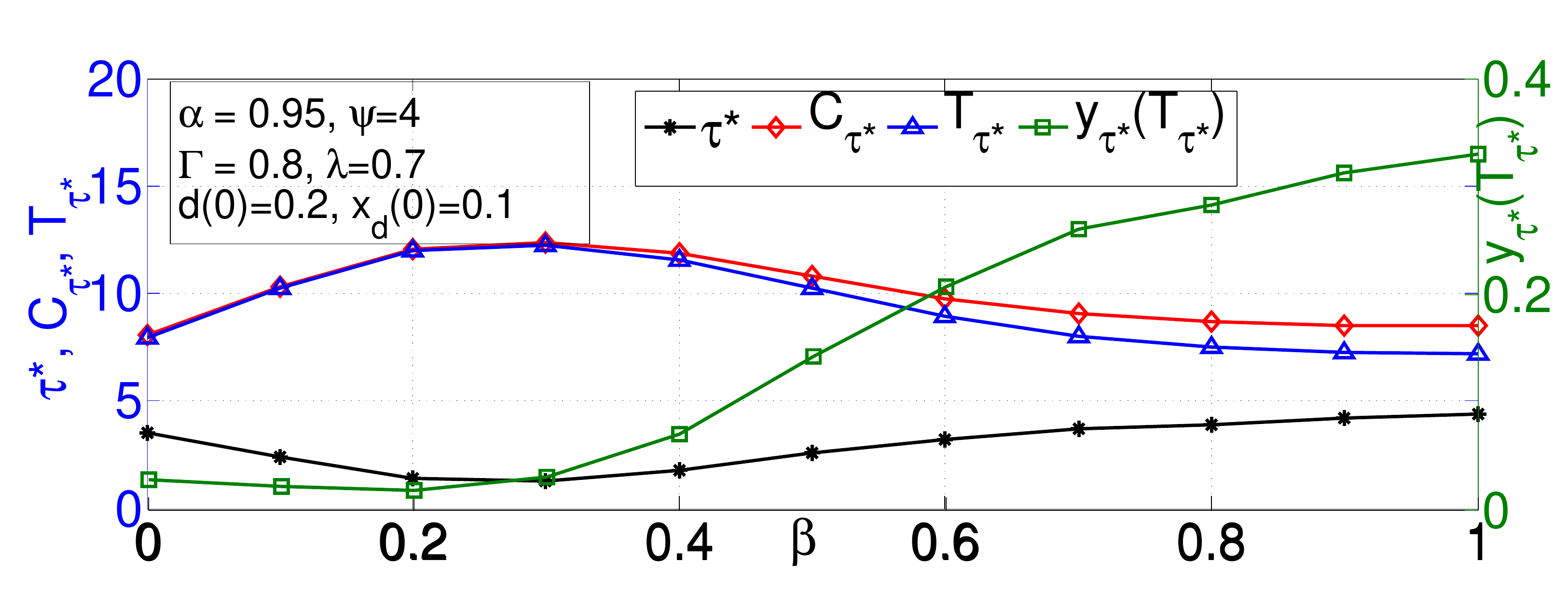}
 \caption{Decentralized Influence Spread: Effect of $\beta$ on the optimal cost and optimal $\tau$}
 \label{fig:beta_sweep}
\end{figure}

\begin{figure}[t]
 \includegraphics[width = 8cm, height = 4cm]{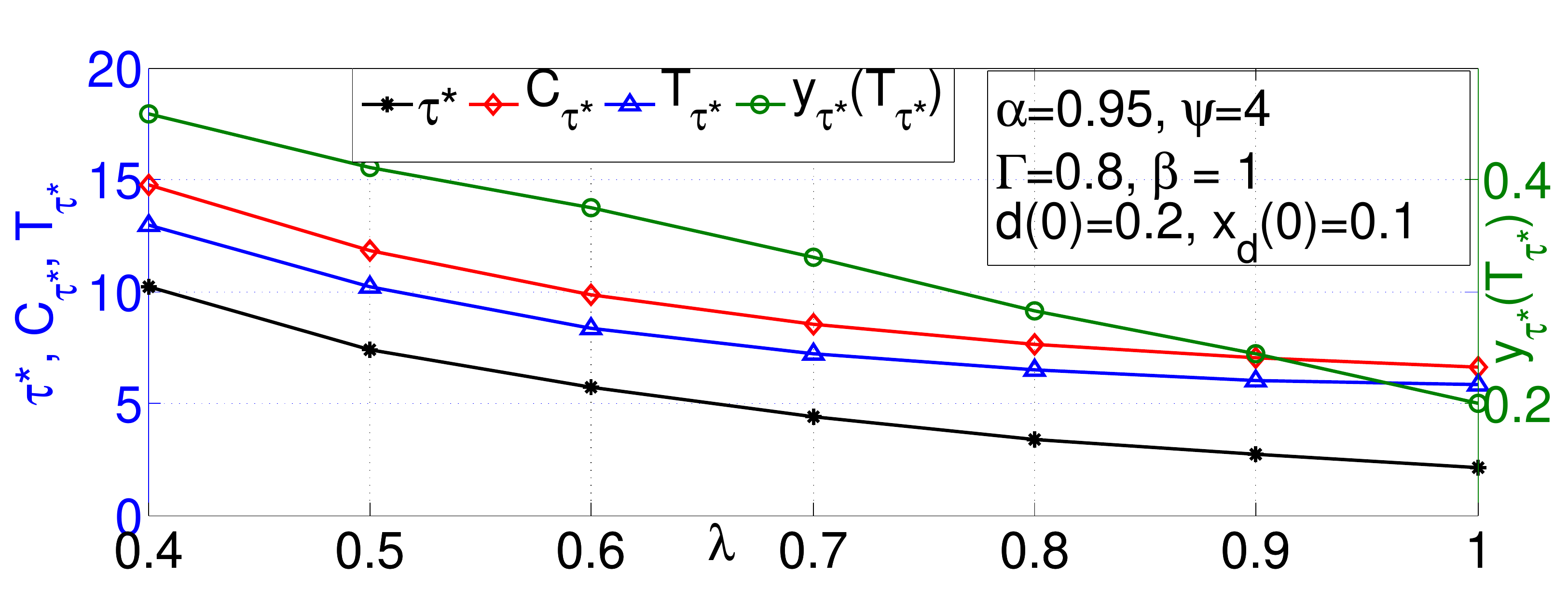}
 \caption{Decentralized Influence Spread: Effect of $\lambda$ on the optimal cost and optimal $\tau$}
 \label{fig:lambda_sweep}
\end{figure}

\subsubsection{Effect of cost parameters}
Figure~(\ref{fig:alpha_sweep}) shows the effect of $\alpha$ on $\tau^\star$.  For the given system parameters, $a_\infty=0.4352$  and hence it is meaningful to start from $\alpha=0.3$. As $\alpha$ increases, we see that $\tau^\star$ increases and in turn increases the costs, since we need to continue copying to relays for a longer duration. 
Figure~(\ref{fig:psi_sweep}) shows the effect of $\psi$ on $\tau^\star$. As $\psi$ increases, we see that the emphasis on the wasted copies, $y(T)$, increases, and hence the control needs to be \emph{less aggressive}. Thus we see a decrease in $\tau^\star$ as $\psi$ increases, and this leads to a decrease in the wasted copies ($y_{\tau^\star}(T_{\tau^\star})$) and an increase in the delay of reaching the target ($T_{\tau^\star}$). We see an increase in the total cost $C_{\tau^\star}$ since both the terms in the cost function ($T_{\tau^\star}$ and $\psi y_{\tau^\star}(T_{\tau^\star})$) are increasing. 

\begin{figure}[t]
 \includegraphics[width = 8cm, height = 4cm]{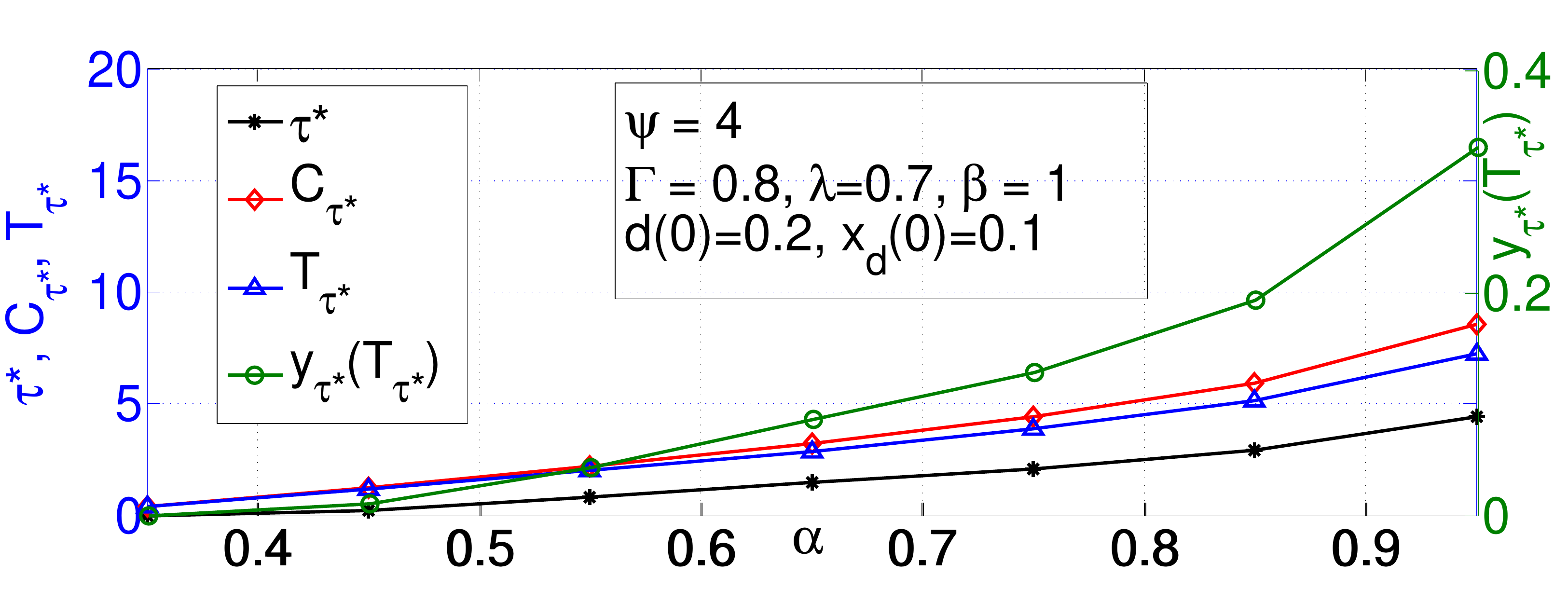}
 
 \caption{Decentralized Influence Spread: Effect of $\alpha$ on the optimal cost and optimal $\tau$}
\label{fig:alpha_sweep}
\end{figure}

\begin{figure}[t]
 \includegraphics[width = 8cm, height = 4cm]{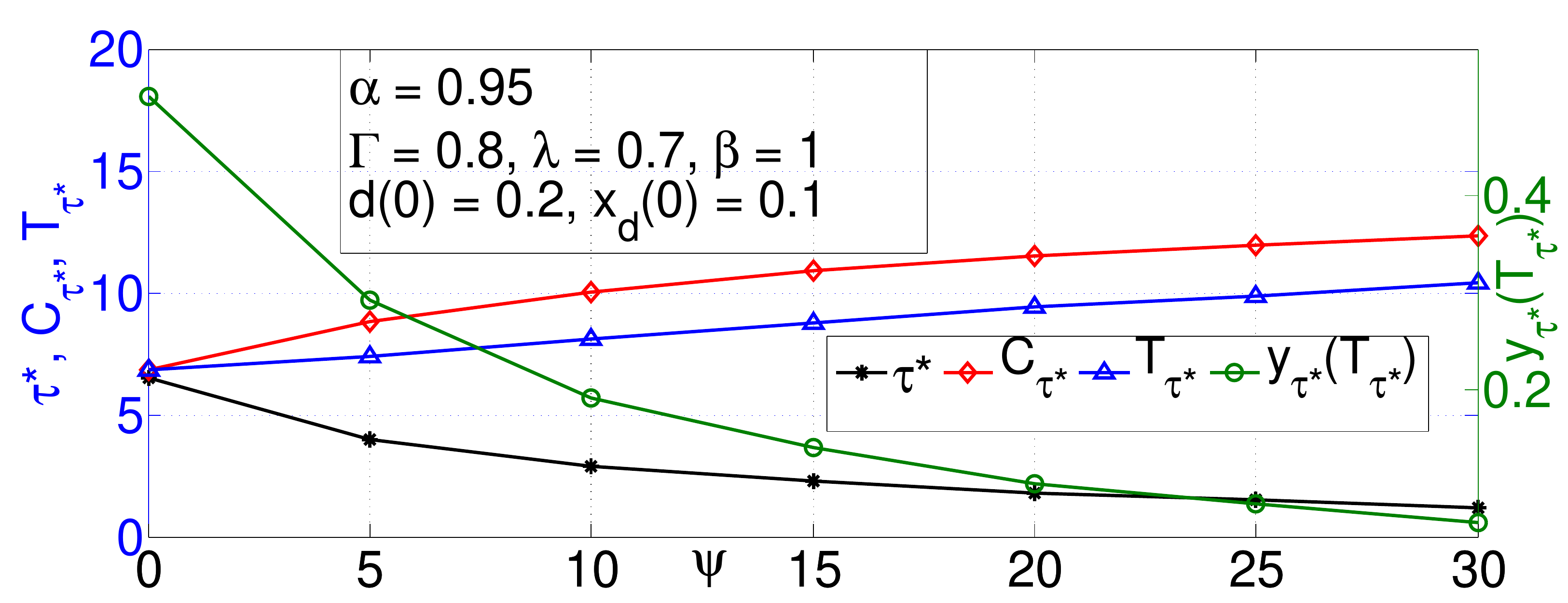}
 
 \caption{Decentralized Influence Spread: Effect of $\psi$ on the optimal cost and optimal $\tau$}
\label{fig:psi_sweep}
\end{figure}

\section{Conclusions}
In this paper we studied the joint evolution of popularity and spread of content in a mobile opportunistic setting. We studied the evolution of popularity under a threshold based model (HILT model) and derived its fluid limits. We also showed that under the exponential distribution for the threshold, the fluid limit of the HILT model is a special case of the SIR model. We then develop a fluid limit for the combined spread of interest and the content(the SIR-SI model). Finally we showed that a time-threshold policy is an optimal copying policy for the SIR-SI model, to optimize the combined cost of target time and the amount of wasted copies.  We also have reported several interesting insights into the evolution of popularity and the co-evolution of popularity and content delivery, which will help the content producers and distributors understand the interplay of various system parameters. 

\appendices 

\section{}
\label{app:HILT-markov}
\begin{proposition} \label{prop:hilt_dtmc}
  For the HILT model, $(B(k), D(k)), k \geq 0,$ is a discrete time Markov chain (DTMC).
\end{proposition}
\proof Since, for each $k \geq 0$, $B(k+1) = B(k) + D(k)$, it suffices to show that, for $k \geq 0$, 
\begin{eqnarray*}
  P(D(k+1) = \ell | (B(0), D(0)), \cdots, ((B(k), D(k))= (b,d)))
\end{eqnarray*}
is a function only of $(b,d)$. We need the following simple lemma (in which the new notation is local to the lemma).

\begin{lemma} \label{lem:recursion_of_threshold_cdf}
  Given $(X_1, X_2, \cdots, X_n)$ i.i.d.\ random variables with common c.d.f.\ $G(x), x \geq a,$ where $a > 0$, and given $b > a$, define
  \begin{eqnarray*}
    \mathcal{Z} = \{i: 1 \leq i \leq n, X_i > b\}, \ \mathrm{and} \ Z = |\mathcal{Z}|.
  \end{eqnarray*}
  Index the indices in $\mathcal{Z}$ as $(i_1, i_2, \cdots, i_Z)$, and  let $Y_1 = X_{i_1}, Y_2 = X_{i_2}, \cdots, Y_Z = X_{i_Z}$ . Then
  \begin{eqnarray*}
    P(Y_i \leq y_i, 1 \leq i \leq Z | Z = m) = \Pi_{i=1}^m \frac{G(y_i) - G(b)}{1 - G(b)}
  \end{eqnarray*}
  for $y_i \geq b, 1 \leq i \leq m.$
\end{lemma}
\proof (of Lemma~\ref{lem:recursion_of_threshold_cdf}) It is easily seen that, for $y_i \geq b, 1 \leq i \leq m,$
\begin{eqnarray*}
\lefteqn{P(Z=m, Y_1 \leq y_1, \cdots, Y_m \leq y_m)} \\
&=& {n \choose m} \  (G(b))^{n-m} \  \Pi_{i=1}^m (G(y_i) - G(b))\\
\end{eqnarray*}
and
\begin{eqnarray*}
  P(Z=m) = {n \choose m} \  (G(b))^{n-m} \  (1 - G(b))^m
\end{eqnarray*}
from which the desired result immediately follows. $\blacksquare$

Returning to the proof of Proposition~\ref{prop:hilt_dtmc}, given $(B(0), D(0)), \cdots, ((B(k), D(k))= (b,d))$, via a recursive application of Lemma~\ref{lem:recursion_of_threshold_cdf}
(starting from the initial i.i.d.\ thresholds $\Theta_i, 1 \leq i \leq N,$ with common c.d.f.\ $F(\theta)$), we conclude that the thresholds of the $m := N - (b+d)$ users in $\mathcal{N} \ \backslash \
(\mathcal{B}(k) \cup \mathcal{D}(k))$ are i.i.d.\ with common c.d.f.\ $\frac{F(\theta) - F(\gamma j)}{1 - F(\gamma j)}$, over the range $\theta \geq \gamma j$. At the end of period $k$, the newly influenced
users in $\mathcal{D}(k)$ will serve as additional influence on the as yet uninfluenced users. Defining, 
\[ p_\gamma(b,d) = \left(\frac{F(\gamma(b+d)) - F(\gamma b)}{1 -  F(\gamma b)}\right)\]

Of the users in $\mathcal{N} \ \backslash \ (\mathcal{B}(k) \cup \mathcal{D}(k))$, $\ell$ will become interested (i.e., $D(k+1) = \ell$) with probability 
\begin{eqnarray*}
  {N - (b+d) \choose \ell} 
  p_\gamma(b,d)^\ell
  (1 - p_\gamma(b,d))^{(N - (b+d)) - \ell}
\end{eqnarray*}
which depends on the ``history'' only via $(b,d)$, and thereby establishes the desired result.  $\blacksquare$

\section{Fluid Limit for the HILT Model}
\label{app:hilt-fluidlimit}
Recalling from Section~\ref{sec:HILT-model}, the evolution of $(\Tilde{B}^{N}(t), \Tilde{D}^{N}(t))$ can be written in terms of its mean ``drift'' at any $t$ as follows.
\[ \Tilde{B}^{N}(t+1) = \Tilde{B}^{N}(t) + \frac{\Tilde{D}^{N}(t)}{N} + \Tilde{Z}^{N}_b(t+1) \]
\begin{eqnarray*}
\lefteqn{\Tilde{D}^{N}(t+1) = \frac{N-1}{N} \Tilde{D}^{N}(t)}\\
\hspace{-2cm} &+& \mathbb{E} \bigg[ \frac{F(\gamma_N N (\Tilde{B}^{N}(t) + \Tilde{C}^{N}(t)))- F(\gamma_N N (\Tilde{B}^{N}(t))}{1-F(\gamma_N N (\Tilde{B}^{N}(t))} \bigg]\times\\
& & (1-\Tilde{B}^{N}(t) -\Tilde{D}^{N}(t)) + \Tilde{Z}^{N}_d(t+1)
\end{eqnarray*}
where we have defined $\Tilde{Z}^{N}_b(t)$ and $\Tilde{Z}^{N}_d(t)$ in an analogous manner. The mean drift rate function, $g^{N}(\Tilde{B}^{N}(k),\Tilde{D}^{N}(k))$ per unit update step size $\frac{1}{N}$ becomes,
\[\frac{g^{N}(\Tilde{B}^{N}(t),\Tilde{D}^{N}(t))}{\frac{1}{N}} = \bigg( g^{N}_1(.), g^{N}_2(.) \bigg) \]
where 
\[g^{N}_1(.)= \Tilde{D}^{N}(t) \]
\begin{eqnarray*}
\lefteqn{g^{N}_2(.)} \\
&=&  N \mathbb{E} \bigg[ \frac{F(\gamma_N N (\Tilde{B}^{N}(t) + \Tilde{C}^{N}(t)))- F(\gamma_N N (\Tilde{B}^{N}(t))}{1-F(\gamma_N N (\Tilde{B}^{N}(t))} \bigg] \times \\ 
& &\bigg( 1-\Tilde{B}^{N}(t)- \Tilde{D}^{N}(t) \bigg) - \Tilde{D}^{N}(t) 
\end{eqnarray*}
Recall $\gamma_N \times (N-1) = \Gamma$ and $f(x)$ the density function of the threshold distribution $F$. We can then show that (see Appendix~\ref{app:limit_expectation}), 
\begin{eqnarray*}
\lim_{N \rightarrow \infty} N \mathbb{E} \bigg[ \frac{F(\gamma_N N (\Tilde{B}^{N}(t) + \Tilde{C}^{N}(t)))- F(\gamma_N N (\Tilde{B}^{N}(t))}{1-F(\gamma_N N (\Tilde{B}^{N}(t))} \bigg]\\
= \frac{f(\Gamma b)\Gamma d}{1- F(\Gamma b)} \\
\end{eqnarray*}
Letting $g_1(b,d)= d$ and $g_2(b,d) =  \frac{f(\Gamma b)\Gamma d}{1- F(\Gamma b)} (1-b-d) - d$ and define,
\[ g(b,d) := \bigg( g_1(b,d),g_2(b,d) \bigg) \]
We can then use these limiting drift functions to obtain the fluid limits corresponding to the HILT model.

\label{app:limit_expectation}
Consider 
\[ lim_{N \rightarrow \infty} N \mathbb{E} \bigg[ \frac{F(\gamma N (\Tilde{B}^{N}(t) + \Tilde{C}^{N}(t)))- F(\gamma N (\Tilde{B}^{N}(t))}{1-F(\gamma N (\Tilde{B}^{N}(t))} \bigg]\] 
where the expectation is with respect to $\Tilde{C}^N(t)$ given $(\Tilde{B}^{N}(t), \Tilde{D}^{N}(t))$. Applying Taylor's expansion to $F(\gamma N (\Tilde{B}^{N}(t) + \Tilde{C}^{N}(t)))$ around $\Tilde{B}^{N}(t)$ we get,

\begin{eqnarray*}
lim_{N \rightarrow \infty}N \mathbb{E} \bigg[ \frac{F(\gamma N (\Tilde{B}^{N}(t) + \Tilde{C}^{N}(t)))- F(\gamma N (\Tilde{B}^{N}(t))}{1-F(\gamma N (\Tilde{B}^{N}(t))} \bigg]\\
= lim_{N \rightarrow \infty} N \mathbb{E} \bigg[ \frac{ \gamma N \Tilde{C}^N(t) f(\gamma N \Tilde{B}^N(t))}{1-F(\gamma N (\Tilde{B}^{N}(t))} \bigg]\\ + lim_{N \rightarrow \infty} N \mathbb{E} \bigg[ \frac{ \gamma N \frac{\Tilde{C}^N(t)^2}{2} \dot{f} (\zeta)}{1-F(\gamma N (\Tilde{B}^{N}(t))} \bigg] \\
\end{eqnarray*}

In the second term since $\Tilde{C}^N(t) = \frac{C^N(t)}{N}$ where $C^N(t) \sim \mathtt{Bin} (D^N(t), \frac{1}{N})$, we have

\begin{eqnarray*}
\mathbb{E}(\Tilde{C}^N(t)^2) &=& \frac{1}{N^2} \mathbb{E} (C^N(t)^2)\\
&=& \frac{1}{N^2} \bigg[ D^N(t) \frac{1}{N} (1-\frac{1}{N}) + (\frac{D^N(t)}{N})^2 \bigg]\\
&=& \frac{1}{N^2} \bigg[ \Tilde{D}^N(t) (1-\frac{1}{N}) + (\Tilde{D}^N(t))^2 \bigg]\\
\end{eqnarray*}

Also $\dot{f} (\zeta)$ is bounded, since we require the drift functions to satisfy the Lipschitz condition. Hence the second term vanishes as $N \to \infty$. In the first term, noting that $\gamma N \to \Gamma$ and $\mathbb{E} (\Tilde{C}^N(t)) = \frac{\Tilde{D}^N(t)}{N}$, given $\Tilde{B}^N(t)=b$, $\Tilde{D}^N(t) =d$, the limit becomes $\frac{\Gamma d f(\Gamma b)}{1-F(\Gamma b)}$.

\label{app:kurtz-applied-hilt}
Kurtz's theorem \cite{kurtz70ode-markov-jump-processes} provides us a way by which we can, subject to certain conditions, approximate the evolution of a pure jump Markov process by the solution of a derived ODE. Along the lines of Theorem 2.8 in \cite{darling02limits-purejump-markov},to prove the theorem \ref{thm:kurtz-theorem-hilt} we need to verify the following four equivalent conditions for the HILT model.
Let $g_1(b,d)= d$ and $g_2(b,d) =  \frac{\mathbf{f}(\Gamma b)\Gamma d}{1- F(\Gamma b)} (1-b-d) - d$ and define,
\[ g(b,d) := \bigg( g_1(b,d), g_2(b,d) \bigg) \]

We require
\begin{itemize}
 \item Lipschitz property of $g(b,d)$
 \item Uniform convergence of $\frac{g^{(N)}(b,d)}{\frac{1}{N}}$ to $g(b,d)$
 \item Bounded Noise variance 
 \item Convergence of initial conditions
\end{itemize}

\begin{itemize}
\item[(i)]\textbf{Lipschitz property}
We have,
\[ \frac{\partial g_1}{\partial b} = 0 ;  \frac{\partial f_1}{\partial d} =  1\]

\[ \frac{\partial g_2}{\partial b} = \Gamma^2 d \dot{h}_F (\Gamma b) (1-b-d) - \Gamma d h_F(\Gamma b)\]
\[ \frac{\partial g_2}{\partial d} = h_f(\Gamma b) \Gamma (1-b-2d) -1 \]
Under the assumption that $\dot{f}(\cdot)$ is bounded, we see that each of the terms above is bounded when $(b,d) \in [0,1]\times[0,1-b]$. Thus the norm of Jacobian $|| Dg(b,d) ||$ is uniformly bounded, and it follows that $g(b,d) = (g_1(b,d),g_2(b,d))$ is Lipschitz.

\item[(ii)]\textbf{Uniform Convergence}
Let $g^N(b,d)$ be as defined in Appendix~\ref{app:hilt-fluidlimit}. By definition, $\lim_{N \rightarrow \infty} N \gamma =\Gamma$ and by the steps in Appendix~\ref{app:limit_expectation}, the uniform convergence of $\frac{g^{N}(b,d)}{1/N}$ to $g(b,d)$ in the domain $(b,d) \in [0,\frac{1}{N},\cdots ,N]\times[0,\frac{1}{N}, \cdots,b]$ is proven. 

\item[(iii)]\textbf{Bounded Noise variance}
We will write the noise variance terms for the scaled processes $B^N(k)$ and $D^N(k)$, and from them derive the noise variance terms for $\Tilde{B}^N(k)$ and $\Tilde{D}^N(k)$ by scaling them by $\frac{1}{N^2}$. 
Let $\mathcal{Z}^{N}_b(k)$ and $\mathcal{Z}^{N}_d(k)$ be the zero mean random variables representing the noise in the drift terms of $B^N(k)$ and $D^N(k)$ respectively. We can then write,
\[ E[ | \mathcal{Z}^{N}_b(k)|^2 | \mathcal{F}_k] = (1-\frac{1}{N})\frac{1}{N} D^N(k) \]
\begin{eqnarray*}
\lefteqn{ E[ | \mathcal{Z}^{N}_d(k)|^2 | \mathcal{F}_k]}\\
 &=& (1-\frac{1}{N})\frac{1}{N} D^N(k)\\
 &+& \sum_{z=0}^{z=D^N(k)} \bigg[ g(z) (1-g(z)) (N- B^N(k) - D^N(k))\\
 & & \times {D^N(k) \choose z} \left(\frac{1}{N}\right)^z \left(1-\frac{1}{N}\right)^{(D^N(k)-z)} \bigg]\\
\end{eqnarray*}
where, 
\[ g(z) = \frac{F(\gamma N (B^N(k)+z)) - F(\gamma N B^N(k))}{1 - F(\gamma N B^N(k))} \]
Note that both these terms can be upper bounded by constants say $c_b$ and $c_d$, observing that  $(b,d) \in [0,1]\times[0,1-b]$ and hence the noise conditions for $\Tilde{B}^N(k)$ and $\Tilde{D}^N(k)$ are satisfied.

\item[(iv)]\textbf{Convergence of initial conditions}
By choice, we have $\Tilde{B}^{N}(0) = b(0)$ and $\Tilde{D}^{N}(0)=d(0)$. 
\end{itemize}
The above four conditions validate the application of Kurtz's theorem and thus proves Theorem \ref{thm:kurtz-theorem-hilt}.
\begin{flushright}
$\blacksquare$
\end{flushright}

\section{Verification of Kurtz's theorem conditions for SIR-SI model}
\label{app:kurtz-applied-sirsi}
The mean drift rates for the SIR-SI model can be given by,
\begin{eqnarray*}
F^N_{b}(k) &=& \beta_N \Tilde{D}(k)\\
F^N_{d}(k) &=& - \beta_N \Tilde{D}(k) + N \lambda_N \Gamma \Tilde{D}(k) \Tilde{S}(k) \\
F^N_{x_b}(k) &=& \beta_N \Tilde{X}_d(k) + \\
&& N \lambda_N (\Tilde{B}(k) - \Tilde{X}_b(k))(\Tilde{X}(k) + \Tilde{Y}(k)) \\
F^N_{x_d}(k) &=& \Gamma N \lambda_N (\Tilde{D}(k) - \Tilde{X}_d(k))\Tilde{Y}(k) - \\
&&  \beta_N \Tilde{X}_d(k)+ \Gamma N \lambda_N \Tilde{X}_d(k) \Tilde{S}(k) + \\ \nonumber
&& N \lambda_N (\Tilde{D}(k) - \Tilde{X}_d(k))(\Tilde{X}(k)+\Tilde{Y}(k))\\
F^N_{y}(k) &=& -\Gamma N \lambda_N \Tilde{D}(k) \Tilde{Y}(k) + \\
&& \hspace{-1cm} N \lambda_N \sigma (\Tilde{S}(k)-\Tilde{Y}(k))(\Tilde{X}_b(k) +\Tilde{Y}(k) +(1-\Gamma)\Tilde{X}_d(k)) 
\end{eqnarray*}
and the corresponding o.d.e. equations, 

\begin{eqnarray}
\dot{b} &=& \beta d\\
\dot{d} &=& - \beta d + \lambda \Gamma d s \\
\dot{x_b} &=& \beta x_d + \lambda (b - x_b)(x+y) \\
\dot{x_d} &=& \Gamma \lambda (d - x_d)y + \Gamma \lambda x_d s + \\ \nonumber
&& \hspace{2cm} \lambda (d-x_d)(x+y) - \beta x_d \\
\dot{y} &=& -\Gamma \lambda d y + \lambda \sigma (s-y)(x_b +y +(1-\Gamma)x_d) 
\end{eqnarray}

Let $f_b, f_d, f_{x_b}, f_{x_d}, f_y$ denote the right hand sides of the above fluid limit equations and let $f := (f_b, f_d, f_{x_b}, f_{x_d}, f_y)$. Also let $F^N(k) := (F^N_b(k), F^N_d(k), F^N_{x_b}(k), F^N_{x_d}(k), F^N_y(k))$.

\begin{itemize}
\item[(i)]\textbf{Lipschitz property}
We see that each of the partial derivatives $\frac{\partial f_u}{\partial v}$ where $u,v \in (b,d,x_b,x_d,y)$ is bounded when $(b,d,x_b,x_d,y) \in [0,1]\times[0,1-b]\times[0,b]\times[0,d]\times[0,1-b-d]$. Thus the norm of Jacobian is uniformly bounded, and it follows that $f(z)$ is Lipschitz.

\item[(ii)]\textbf{Uniform Convergence}
Taking $N \lambda_N \rightarrow \lambda$ and $\beta_N \rightarrow \beta$ we see that the uniform convergence of $F^N(z)$ to $f(z)$ is straightforward. 

\item[(iii)]\textbf{Bounded Noise variance}
Since in the co-evolution system the maximum jump rate out of a state is bounded $\lambda_N N(N-1) + \beta_N N$, and since $N \lambda_N \rightarrow \lambda$ and $\beta_N \rightarrow \beta$, the jump rate is $O(N)$; the increments for the scaled process $\Tilde{Z}^N(.)$ are $O(N^{-1})$. This is referred to as ``hydrodynamic scaling'' \cite{darling02limits-purejump-markov}. This ensures that the noise convergence to zero is ensured.

\item[(iv)]\textbf{Convergence of initial conditions}
The initial conditions are chosen such that $(\Tilde{B}^N(0), \Tilde{D}^N(0), \Tilde{X}_b^N(0), \Tilde{X}_d^N(0), \Tilde{Y}^N(0)) = (b(0), d(0), x_b(0), x_d(0), y(0))$
\end{itemize}

\section{}
\label{app:coop-system}
Let $\mathbf{w}(t)$ be the solution of the o.d.e $\dot{\mathbf{w}}=f(\mathbf{w};z)$ with piecewise Lipschitz continuous control $z(t)$, where $f$ is continuously differentiable and Lipschitz in $\mathbf{w}$ and $z$. Let $\mathbf{w}^{(1)}(t)$ and  $\mathbf{w}^{(2)}(t)$ be the trajectories corresponding to two controls $z^{(1)}(t)$ and $z^{(2)}(t)$ respectively, i.e.,
\begin{eqnarray}
\label{eqn:w_1} 
\dot{\mathbf{w}}^{(1)}&=&f(\mathbf{w}^{(1)};z^{(1)})\\
\label{eqn:w_2}
\dot{\mathbf{w}}^{(2)}&=&f(\mathbf{w}^{(2)};z^{(2)})
\end{eqnarray}

Motivated by the Kamke condition (see \cite{smith08monotone}) we define Kamke-dominance between two controls $z^{(1)}(t)$ and $z^{(2)}(t)$ as follows.

\begin{definition}
\label{defn:control-domination}
We say \emph{ $z^{(1)}(t)$ Kamke-dominates $z^{(2)}(t)$ in the system $f(\mathbf{w};z)$} if for each $t$ where $\mathbf{w}^{(1)}(t_0) \geq \mathbf{w}^{(2)}(t_0)$ and $\exists i$ with $w^{(1)}_{i}(t_0) = w^{(2)}_{i}(t_0)$ then 
\[ f_i(\mathbf{w}^{(1)}(t_0);z^{(1)}(t_0)) \geq f_i(\mathbf{w}^{(2)}(t_0);z^{(2)}(t_0))\]
\end{definition}

\begin{lemma}
\label{lemma:order-preserving}
Let $\mathbf{w}^{(1)}(t)$ and  $\mathbf{w}^{(2)}(t)$ be as in Equations~(\ref{eqn:w_1}) and (\ref{eqn:w_2}), and $z^{(1)}(t)$ Kamke-dominate $z^{(2)}(t)$ in the system $f(\mathbf{w};z)$. If $\mathbf{w}^{(1)}(0) \geq \mathbf{w}^{(2)}(0)$ then $\forall t$, $\mathbf{w}^{(1)}(t) \geq \mathbf{w}^{(2)}(t)$.
\end{lemma}

\begin{proof}
The proof of Proposition 1.1 in \cite{smith08monotone} can be adopted directly in writing this proof. Since $f$ is continuously differentiable and Lipschitz in $\mathbf{w}$ and $z$, by the Cauchy-Lipschitz condition, we have a unique solution given the initial conditions. Let $\phi^{(1)}_t (\mathbf{w}^{(1)}(0))$ denote the solution of (\ref{eqn:w_1}) with initial condition $\mathbf{w}^{(1)}(0)$, and $\phi^{(2)}_t (\mathbf{w}^{(2)}(0))$ denote the solution of (\ref{eqn:w_2}) with initial condition $\mathbf{w}^{(2)}(0)$. We wish to show that
$\phi^{(1)}_t (\mathbf{w}^{(1)}(0)) \geq \phi^{(2)}_t (\mathbf{w}^{(2)}(0))$.
Consider $m$ an integer, and let  $\phi^{(1),m}_t (\mathbf{w}^{(1)}(0))$  be the solution corresponding to
\[ \dot{\mathbf{w}}^{(1)}=f(\mathbf{w}^{(1)};z^{(1)}) + \big(\frac{1}{m}\big)\mathbf{e} \]
Then, by continuity of o.d.e. solutions with respect to drift and initial conditions \cite[Chap.1, Lemma 3.1]{hale80ode-book},\ $\phi^{(1),m}_s (\mathbf{w}^{(1)}(0) + \frac{\mathbf{e}}{m})$ defined on $0 \leq s \leq t$ for all large m, say $m>M$ and
\begin{equation}
\label{eqn:ode-continuity}
\phi^{(1),m}_s (\mathbf{w}^{(1)}(0) + \frac{\mathbf{e}}{m}) \rightarrow \phi^{(1)}_s (\mathbf{w}^{(1)}(0))  
\end{equation}
as $m \rightarrow \infty$, uniformly in $s \in [0,t]$. We claim that for $0 \leq s \leq t$, for all $m>M$,
\begin{eqnarray}
\label{eqn:strict_inequal_flows}
 \phi^{(1)m}_s (\mathbf{w}^{(1)}(0) + \frac{\mathbf{e}}{m}) >> \phi^{(2)}_s (\mathbf{w}^{(2)}(0))
\end{eqnarray}
where $x >> y$ implies $x_i > y_i$, $\forall i$.
\paragraph*{Proof of claim} Fix $m>M$. We know that $\mathbf{w}^{(1)}(0) + \frac{\mathbf{e}}{m} =\phi^{(1),m}_0 (\mathbf{w}^{(1)}(0) + \frac{\mathbf{e}}{m}) >> \mathbf{w}^{(2)}(0) = \phi^{(2)}_0 (\mathbf{w}^{(2)}(0))$. By continuity of $\phi^{(1)}$ and $\phi^{(2)}$, the claim (\ref{eqn:strict_inequal_flows}) holds for small $s$. If the claim (\ref{eqn:strict_inequal_flows}) were false, $\exists \ 0<t_0 \leq t$ such that (\ref{eqn:strict_inequal_flows}) holds for $0\leq s < t_0$, and an index $i$ such that, 
\[ \big(\phi^{(1),m}_{t_0} (\mathbf{w}^{(1)}(0) + \frac{\mathbf{e}}{m})\big)_i = \big(\phi^{(2)}_{t_0} (\mathbf{w}^{(2)}(0))\big)_i\]
and 
\begin{equation}
\label{eqn:contradiction-flow}
 \frac{\mathrm{d}}{\mathrm{d}s}\bigg|_{s=t_0} \big(\phi^{(1),m}_{s} (\mathbf{w}^{(1)}(0) + \frac{\mathbf{e}}{m})\big)_i \leq \frac{\mathrm{d}}{\mathrm{d}s}\bigg|_{s=t_0}  \big(\phi^{(2)}_{s} (\mathbf{w}^{(2)}(0))\big)_i 
\end{equation}
But since $\forall j\neq i$
\[ \big(\phi^{(1)m}_{t_0} (\mathbf{w}^{(1)}(0) + \frac{\mathbf{e}}{m})\big)_j \geq \big(\phi^{(2)}_{t_0} (\mathbf{w}^{(2)}(0))\big)_j\]
and $z^{(1)}(t)$ Kamke-dominates $z^{(2)}(t)$ in the system $f(\mathbf{w},z)$, we have
\begin{eqnarray*}
 f_i(\phi^{(1),m}_{t_0} (\mathbf{w}^{(1)}(0) + \frac{\mathbf{e}}{m});z^{(1)}(t_0))+\frac{1}{m}\\
 >  f_i(\phi^{(1),m}_{t_0} (\mathbf{w}^{(1)}(0) + \frac{\mathbf{e}}{m});z^{(1)}(t_0)) \\
\geq f_i(\phi^{(2)}_{t_0} (\mathbf{w}^{(2)}(0));z^{(2)}(t_0))
\end{eqnarray*}
 which implies
\begin{equation*}
 \frac{\mathrm{d}}{\mathrm{d}s}\bigg|_{s=t_0} \big(\phi^{(1),m}_{s} (\mathbf{w}^{(1)}(0) + \frac{\mathbf{e}}{m})\big)_i > \frac{\mathrm{d}}{\mathrm{d}s}\bigg|_{s=t_0}  \big(\phi^{(2)}_{s} (\mathbf{w}^{(2)}(0))\big)_i 
\end{equation*}
This contradicts (\ref{eqn:contradiction-flow}) and hence proves the claim (\ref{eqn:strict_inequal_flows}). Applying (\ref{eqn:ode-continuity}) to this claim completes the proof of the lemma. 
\end{proof} 

Consider the system of equations representing the evolution of $(x_b(t), x_d(t), y(t))$ in the SIR-SI process. 
\begin{eqnarray}
\label{eqn:x_b}
\dot{x_b} &=& \beta x_d + \lambda (b - x_b)(x+y) \\
\label{eqn:x_d}
\dot{x_d} &=& \Gamma \lambda (d - x_d)y + \Gamma \lambda x_d s + \\ \nonumber
&& \hspace{2cm} \lambda (d-x_d)(x+y) - \beta x_d \\
\label{eqn:y}
\dot{y} &=& -\Gamma \lambda d y + \lambda \sigma (s-y)(x_b +y +(1-\Gamma)x_d) 
\end{eqnarray}

Let $w^{(1)} = (x_b^{(1)}, x_d^{(1)}, y^{(1)})$ and  $w^{(2)} = (x_b^{(2)}, x_d^{(2)}, y^{(2)})$ be the solutions respectively for the controls $\sigma^{(1)}(t)$ and $\sigma^{(2)}(t)$ for the above system with $\sigma^{(1)}(t) \leq \sigma^{(2)}(t)$. Define 
\[ (u_b, u_d, v) := (x_b^{(1)}, x_d^{(1)}, y^{(1)}) - (x_b^{(2)}, x_d^{(2)}, y^{(2)})\]
and $\Delta \sigma(t) = \sigma^{(1)}(t) - \sigma^{(2)}(t)$. Then we can write, 
\begin{eqnarray}
\label{eqn:u_b}
\dot{u}_b|_{u_b=0} &=& \beta u_d + \lambda (b-x^{(1)}_b)(u_d + v) \\
\label{eqn:u_d}
\dot{u}_d|_{u_d=0} &=& \Gamma \lambda (d-x_d^{(1)})v + \lambda (d-x_d^{(1)})(u_b + v) \\
\label{eqn:v}
\dot{v}|_{v=0} &=& \lambda (s-y^{(1)}) [ \Delta \sigma y^{(1)} + \sigma^{(1)} u_b + \Delta \sigma x_b^{(2)}  \\\nonumber
& & \hspace{2cm}+(1-\Gamma) (\sigma^{(1)} u_d + \Delta \sigma x_d^{(2)})]
\end{eqnarray}

\begin{lemma} 
\label{lemma:domination-SIRSI}
Control Domination
\begin{enumerate}[(a)]
 \item If $\forall t$, $\sigma^{(1)}(t) \geq \sigma^{(2)}(t)$, then $\sigma^{(1)}(t)$ Kamke-dominates $\sigma^{(2)}(t)$ in the system given by equations~(\ref{eqn:x_b})-(\ref{eqn:y}).
 \item If $\forall t$, $y^{(1)}(t) \geq y^{(2)}(t)$, then $y^{(1)}(t)$ Kamke-dominates $y^{(2)}(t)$ in the system given by equations~(\ref{eqn:x_b})-(\ref{eqn:x_d}). 
\end{enumerate}
\end{lemma}

\begin{proof}
 \begin{enumerate}[(a)]
 \item From equations~(\ref{eqn:u_b})-(\ref{eqn:v}). Since $\Delta \sigma \geq 0$, we find that, 
\begin{itemize}
 \item if $u_d, v \geq 0$ and $u_b=0$, then $\dot{u}_b \geq 0$
 \item if $u_b, v \geq 0$ and $u_d=0$, then $\dot{u}_d \geq 0$ 
 \item if $u_b, u_d \geq 0$ and $v=0$ $\dot{v} \geq 0$
\end{itemize}
This verifies the conditions of Definition~\ref{defn:control-domination}, and thus proves (a).

 \item From equations~(\ref{eqn:u_b})-(\ref{eqn:u_d}). Since $v \geq 0$, we find that, 
\begin{itemize}
 \item if $u_d \geq 0$ and $u_b=0$, then $\dot{u}_b \geq 0$
 \item if $u_b \geq 0$ and $u_d=0$, then $\dot{u}_d \geq 0$ 
\end{itemize}
This verifies the conditions of Definition~\ref{defn:control-domination}, and thus proves (b).
\end{enumerate}
\end{proof}

\begin{lemma} 
\label{lemma:order-preserving-SIRSI}
\begin{enumerate}[(a)]
 \item If $\forall t$, $\sigma^{(1)}(t) \geq \sigma^{(2)}(t)$, and $w^{(1)}(0) \geq w^{(2)}(0)$, then $\forall t$, $w^{(1)}(t) \geq w^{(2)}(t)$ 
 \item If $\forall t$, $y^{(1)}(t) \geq y^{(2)}(t)$, and $(x_b^{(1)}(0), x_d^{(1)}(0)) \geq (x_b^{(2)}(0), x_d^{(2)}(0))$, then $\forall t$, \[(x_b^{(1)}(t), x_d^{(1)}(t)) \geq (x_b^{(2)}(t), x_d^{(2)}(t))\] 
\end{enumerate}
\end{lemma}
\begin{proof}
 The proof follows by applying Lemma~\ref{lemma:order-preserving} to the systems in Lemma~\ref{lemma:domination-SIRSI} (a) and (b).
\end{proof}

Define $\mathbf{w}(t) = (x_b(t), x_d(t), y(t))$
We will now use the above lemmas to prove the claims in Section~\ref{sec:optimal-control}, in order to establish the optimality of a time-threshold policy. Recall that $\alpha$ has been fixed earlier. Recall the definitions of $T_\sigma$ and $y_\sigma(T_\sigma)$. We shall replace them with $T_\tau$ and $y_\tau(T_\tau)$ whenever $\sigma(t) = \sigma_\tau(t)$, a time-threshold policy. 

Consider $\tau$ such that $\tau \geq T_\tau =: \tau^\prime$. Evidently, $T_{\tau^\prime} = \tau^\prime$, and $y_\tau(T_\tau) = y_{\tau^\prime}(T_{\tau^\prime})$, $\forall \tau \geq \tau^\prime$. 

Define $\mathcal{T} = \{ \tau: \tau \geq 0 \}$ and $\underline{\mathcal{T}} = \{ \tau: \tau \leq T_\tau \} \subset \mathcal{T}$. Observe that $\{ \rho: \exists \ \tau \in \mathcal{T},  y_\tau(T_\tau) = \rho\} = \{ \rho: \exists \ \tau \in \underline{\mathcal{T}},  y_\tau(T_\tau) = \rho\}$. Hence we limit our discussion to $\tau \in \underline{\mathcal{T}}$. 

\begin{lemma}
\label{lemma:tau-rho-continuity}
 The set $\{ \rho : \exists \ \tau \in \underline{\mathcal{T}},  y_\tau(T_\tau) = \rho\}$ forms an interval of the form $[0,\rho_{\max}]$.
\end{lemma}
\begin{proof}
Setting $\tau =0$, i.e., $\sigma(t) = 0$ , $\forall t$, we see from Equation~(\ref{eqn:system-end}), since $y(0)=0$, we have $y(t)=0$, $\forall t$. This establishes that $\rho=0$ belongs to the set. It is also easy to observe that $\rho_{\max}$ is achieved by $\tau$ such that $\tau = T_\tau$. By the intermediate value theorem, it now suffices to show that $y_\tau(T_\tau)$ is a continuous function of $\tau$. 

Consider $\tau, \tau+\delta \in \underline{\mathcal{T}}$, i.e., $\tau \leq T_\tau$, $\tau+\delta \leq T_{\tau+\delta}$. Observe from Lemma~\ref{lemma:order-preserving-SIRSI}, for $\delta >0$, $T_{\tau+\delta} \leq T_\tau$. With the above constraints, evidently the only case to be considered is 
\[\tau < \tau+\delta \leq T_{\tau+\delta} \leq T_\tau\]

Let $\mathbf{w}_\tau(.)$ and $\mathbf{w}_{\tau+\delta}(.)$ be the trajectories corresponding to $\sigma_\tau(t)$ and $\sigma_{\tau+\delta}(t)$, with $\mathbf{w}_\tau(0) = \mathbf{w}_{\tau+\delta}(0)$. Since $\sigma_\tau(t) = \sigma_{\tau+\delta}(t)$, $0  \leq t \leq \tau$, it follows that $\mathbf{w}_\tau(\tau) = \mathbf{w}_{\tau+\delta}(\tau)$. Further, from Equations~(\ref{eqn:system-mid})-(\ref{eqn:system-end}) we observe that, 

$||\mathbf{w}_\tau(\tau+\delta) - \mathbf{w}_{\tau+\delta}(\tau+\delta)||_2 \leq C_1 \delta$, where
\[C_1 = \sqrt{(\beta+\lambda)^2 + (2 \Gamma \lambda + \lambda + \beta)^2 + (\lambda + \Gamma \lambda)^2}\]
$\forall t>\tau+\delta$, $\sigma_\tau(t) = \sigma_{\tau+\delta}(t)$ and hence by continuity of o.d.e. system w.r.t initial conditions at $\tau+\delta$ implied by the Lipschitz nature of $f$ w.r.t $\mathbf{w}$ (see \cite{borkar08stoch-approx}),  $\mathbf{w}_\tau(.)$ is continuous w.r.t $\tau$. 

Recall that $\tau < \tau+\delta \leq T_{\tau+\delta} \leq T_\tau$. By the just observed continuity, for $\epsilon > 0$, we can obtain $\delta > 0$ such that, $ |x_{\tau+\delta} (T_{\tau+\delta}) - x_\tau(T_{\tau+\delta})| \leq \epsilon$ , i.e., 

\[ |\alpha a(\infty) - x_\tau(T_{\tau+\delta})| \leq \epsilon \]

However over the interval $(T_{\tau+\delta}, T_\tau]$, the rate of increase of $x_\tau(.)$ is bounded below by $\lambda a(t)^2 (1-\alpha)\alpha$ (from Equations~(\ref{eqn:system-mid})-(\ref{eqn:system-end})). Hence,

\[ |T_{\tau+\delta} - T_\tau| \leq \frac{\epsilon}{\lambda a(t)^2 (1-\alpha)\alpha} \]
which can be made small by choosing an appropriate $\delta > 0$. We have thus shown that $T_\tau$ continuous w.r.t $\tau$. $\tau \in \underline{\mathcal{T}}$.

To show the continuity of $y_\tau(T_\tau)$, consider,
\begin{eqnarray*}
 \lefteqn{| y_\tau(T_\tau) - y_{\tau+\delta}(T_{\tau+\delta}) |  \leq}  && \\
 & & | y_\tau(T_\tau) - y_{\tau}(T_{\tau+\delta}) | + | y_\tau(T_{\tau+\delta}) - y_{\tau+\delta}(T_{\tau+\delta}) |
\end{eqnarray*}
In the above equation, the first term on the right hand side can be made arbitrarily small by the continuity of $T_\tau$ w.r.t $\tau$ and the fact that $y_\tau(.)$ is a continuous trajectory, and the second term can be made arbitrarily small by the continuity of $\mathbf{w}_\tau(.)$ w.r.t $\tau$. Thus $y_\tau(T_\tau)$ is a continuous function of $\tau$, in the space of threshold policies. And since $y_\tau(T_\tau):\tau \rightarrow \rho$ and $\tau \in [0, \infty)$, we see that $\rho \in [0,\rho_{max}]$.
\end{proof}

\begin{lemma}
\label{lemma:threshold-dominates-samerho}
Let $\sigma(t)$ be a policy such that  $y_\sigma(T_\sigma) < \rho_{\max}$. Then there exists a threshold policy $\sigma_\tau(t)$ whose cost is no worse than that of $\sigma(t)$.
\end{lemma}

\begin{proof}
Since  $\rho := y_\sigma(T_\sigma) < \rho_{\max}$, there exists $\tau \geq 0$ such that $y_\tau(T_{\tau}) = \rho$. We will argue that $\sigma_\tau(t)$ is such that $T_\tau \leq T_\sigma$. 

Let $\mathbf{w}_\tau(t)$ and $\mathbf{w}_\sigma(t)$ be the system trajectories for the controls $\sigma_\tau(t)$ and $\sigma(t)$ respectively with $\mathbf{w}_\tau(0)=\mathbf{w}_\sigma(0)$. By definition, for $t \leq \tau$, 
$\sigma_\tau(t) \geq \sigma(t)$ and hence, by Lemma~\ref{lemma:order-preserving-SIRSI}, $\mathbf{w}_\tau(\tau)  \geq \mathbf{w}_\sigma(\tau) $. 

Suppose, contrary to the claim, $T_\tau > T_\sigma$. Evidently, this cannot happen if for all $t$, $0 \leq t \leq T_{\sigma}$, $y_\tau(t)  \geq y_\sigma(t)$. This is because, by Lemma~\ref{lemma:order-preserving-SIRSI}, we will have $x_\tau(t) \geq x_\sigma(t)$ and hence $T_\tau \leq T_\sigma$.

Hence, there exists $t_0$, $\tau \leq t_0 < T_\sigma$, such that $y_\tau(t_0) = y_\sigma(t_0)$. Then we have $y_\tau(t_0) = y_\sigma(t_0) > \rho$ (since $y_\tau(T_\tau)=\rho$ and $\tau \leq t_0 < T_\tau$, and for $t \geq \tau$, $\dot{y}_\tau(t) < 0$).
Also, for $t\geq t_0$, $\sigma_\tau(t) = 0 \leq \sigma(t)$. Thus from (\ref{eqn:system-end}), $\forall s \in \{ s: s> t_0, y_\tau(s) = y_\sigma(s)\}$, $\dot{y}_\tau(s) \leq \dot{y}_\sigma(s)$. Thus for $s\geq t_0$, $y_\tau(s) \leq y_\sigma(s)$. 
Since $y_\tau(T_\tau) = y_\sigma(T_\sigma) = \rho$, this implies $T_\tau \leq T_\sigma$.

\end{proof}

\begin{lemma}
\label{lemma:rho-greater-than-rhomax}
 Consider a non-threshold policy $\sigma(t)$ which achieves $\rho > \rho_{max}$. Consider the time-threshold policy $\hat{\sigma}(t)$ of the form, 
\begin{eqnarray*}
 \hat{\sigma}(t) = \left\{\begin{array}{cl}
                              1, & 0 \leq t < \sup\{t:\sigma(t) > 0\} \\
			      0, & \mbox{otherwise}
			      \end{array}\right.  
\end{eqnarray*}
Then $\hat{\sigma}(t)$ policy achieves a smaller total cost (given by Equation~\ref{eqn:obj-function}) than $\sigma(t)$.
\end{lemma}

\begin{proof}
Let $\sigma(t)$ and $\sigma_{\hat{\tau}}(t)(t)$ be as above. Let $\mathbf{w}_{\hat{\tau}}(t)$ and $T_{\hat{\tau}}$ be as defined earlier for time-threshold policies, corresponding to $\sigma_{\hat{\tau}}(t)$. Then by Lemma~\ref{lemma:tau-rho-continuity},  $y_{\hat{\tau}}(T_{\hat{\tau}}) \leq \rho_{max} < \rho = y_\sigma(T_\sigma)$, and by Lemma~\ref{lemma:order-preserving-SIRSI}, $T_{\hat{\tau}} \leq T_\sigma$, and hence the $\sigma_{\hat{\tau}}(t)$ policy achieves a smaller total cost than $\sigma(t)$.
\end{proof}

\bibliographystyle{IEEEtran}
\bibliography{bib-infocom12}

\begin{thebibliography}{10}
\providecommand{\url}[1]{#1}
\csname url@samestyle\endcsname
\providecommand{\newblock}{\relax}
\providecommand{\bibinfo}[2]{#2}
\providecommand{\BIBentrySTDinterwordspacing}{\spaceskip=0pt\relax}
\providecommand{\BIBentryALTinterwordstretchfactor}{4}
\providecommand{\BIBentryALTinterwordspacing}{\spaceskip=\fontdimen2\font plus
\BIBentryALTinterwordstretchfactor\fontdimen3\font minus
  \fontdimen4\font\relax}
\providecommand{\BIBforeignlanguage}[2]{{%
\expandafter\ifx\csname l@#1\endcsname\relax
\typeout{** WARNING: IEEEtran.bst: No hyphenation pattern has been}%
\typeout{** loaded for the language `#1'. Using the pattern for}%
\typeout{** the default language instead.}%
\else
\language=\csname l@#1\endcsname
\fi
#2}}
\providecommand{\BIBdecl}{\relax}
\BIBdecl

\bibitem{srini-kumar12coevol-mp2p}
V.~Srinivasan and A.~Kumar, ``Co-evolution of content popularity and delivery
  in mobile p2p networks,'' in \emph{In Proceedings of IEEE Infocom 2012
  Mini-conference, Orlando, FL, USA}, 2012.

\bibitem{repantis-kalogeraki04data-dissem-mp2p}
T.~Repantis and V.~Kalogeraki, ``Data dissemination in mobile peer-to-peer
  networks,'' in \emph{Proceedings of the 6th international conference on
  Mobile data management}.\hskip 1em plus 0.5em minus 0.4em\relax ACM, 2005,
  pp. 211--219.

\bibitem{singh-etal11dtn-multi-destination}
C.~Singh, A.~Kumar, R.~Sundaresan, and E.~Altman, ``Optimal forwarding in delay
  tolerant networks with multiple destinations,'' in \emph{9th Intl. Symposium
  on Modeling and Optimization in Mobile, Ad Hoc, and Wireless Networks
  (WiOpt)}, 2011.

\bibitem{khouzani-etal10patch-dissemination}
M.~Khouzani, S.~Sarkar, and E.~Altman, ``Dispatch then stop: Optimal
  dissemination of security patches in mobile wireless networks,'' in
  \emph{49th IEEE Conference on Decision and Control (CDC), 2010}.\hskip 1em
  plus 0.5em minus 0.4em\relax IEEE, 2010, pp. 2354--2359.

\bibitem{shakkottai-etal10demand-aware-content-spread}
S.~Shakkottai and R.~Johari, ``Demand-aware content distribution on the
  internet,'' in \emph{IEEE/ACM Transactions on Networking (TON)}, vol.~18,
  no.~2.\hskip 1em plus 0.5em minus 0.4em\relax IEEE Press, 2010, pp. 476--489.

\bibitem{kempe-etal03max-spread-infl}
D.~Kempe, J.~Kleinberg, and {\'E}.~Tardos, ``Maximizing the spread of influence
  through a social network,'' in \emph{Proceedings of the ninth ACM SIGKDD
  international conference on Knowledge discovery and data mining}.\hskip 1em
  plus 0.5em minus 0.4em\relax ACM, 2003, pp. 137--146.

\bibitem{kermack-mckendrick27SIR-epidemics}
W.~O. Kermack and A.~G. McKendrick, ``A contribution to the mathematical theory
  of epidemics,'' in \emph{Proceedings of the Royal Society of London}, ser. A,
  Containing Papers of a Mathematical and Physical Character, vol. 115, no.
  772.\hskip 1em plus 0.5em minus 0.4em\relax Royal Society, 1927, pp.
  700--721.

\bibitem{srini-kumar11LT-model-ncc}
V.~Srinivasan and A.~Kumar, ``Information dissemination in socially aware
  networks under the linear threshold model,'' in \emph{2011 National
  Conference on Communications (NCC)}, January 2011, pp. 1 --5.

\bibitem{kurtz70ode-markov-jump-processes}
T.~Kurtz, ``Solutions of ordinary differential equations as limits of pure jump
  markov processes,'' in \emph{Journal of Applied Probability}.\hskip 1em plus
  0.5em minus 0.4em\relax JSTOR, 1970, pp. 49--58.

\bibitem{darling02limits-purejump-markov}
R.~Darling, ``Fluid limits of pure jump markov processes: A practical guide,''
  2002, available at arxiv.org/pdf/math/0210109.

\bibitem{benaim-leboudec08mean-field-models}
M.~Bena{\=\i}m and J.~Le~Boudec, ``A class of mean field interaction models for
  computer and communication systems,'' in \emph{Performance Evaluation},
  vol.~65, no. 11-12.\hskip 1em plus 0.5em minus 0.4em\relax Elsevier, 2008,
  pp. 823--838.

\bibitem{cox61renewal-theory}
D.~R. Cox, \emph{Renewal Theory}.\hskip 1em plus 0.5em minus 0.4em\relax
  Metheun \& Co. Ltd. Science Paperbacks, 1961.

\bibitem{granovetter78threshold-models}
M.~Granovetter, ``Threshold models of collective behavior,'' \emph{American
  journal of sociology}, pp. 1420--1443, 1978.

\bibitem{daley-gani99epidemic-modeling}
D.~J. Daley and J.~Gani, \emph{Epidemic Modelling: An Introduction}.\hskip 1em
  plus 0.5em minus 0.4em\relax Cambridge University Press, 2001.

\bibitem{gast10mean-field-mdp}
N.~Gast, B.~Gaujal, and J.~Boudec, ``Mean field for markov decision processes:
  from discrete to continuous optimization,'' \emph{Arxiv preprint
  arXiv:1004.2342}, 2010.

\bibitem{smith08monotone}
H.~Smith, \emph{Monotone dynamical systems: An introduction to the theory of
  competitive and cooperative systems}.\hskip 1em plus 0.5em minus 0.4em\relax
  American Mathematical Soc., 2008.

\bibitem{hale80ode-book}
J.~Hale, \emph{Ordinary Differential Equations}.\hskip 1em plus 0.5em minus
  0.4em\relax RE Krieger, Malabar, FL, 1980.

\bibitem{borkar08stoch-approx}
V.~Borkar, \emph{Stochastic approximation: a dynamical systems
  viewpoint}.\hskip 1em plus 0.5em minus 0.4em\relax Cambridge Univ Pr, 2008.

\end{thebibliography}

\end{document}